\documentclass[11pt]{article}
\usepackage[margin=1in,letterpaper]{geometry}
\usepackage[T1]{fontenc}
\usepackage{lmodern}
\usepackage{microtype}

\usepackage{enumerate} 
\usepackage{amsmath,amsthm, amssymb}
\usepackage{mathtools}
\usepackage{nicefrac}
\usepackage{booktabs}
\usepackage{array}
\usepackage{complexity}

\usepackage[hidelinks]{hyperref}
\usepackage{cleveref}

\newtheorem{theorem}{Theorem}
\newtheorem{lemma}[theorem]{Lemma}
\newtheorem{definition}[theorem]{Definition}
\newtheorem{corollary}[theorem]{Corollary}
\newtheorem*{theorem*}{Theorem}

\DeclareMathOperator{\indeg}{indeg}
\DeclareMathOperator{\SUM}{SUM}

\newcommand{\bigO}{\ensuremath{\mathcal O}}
\newcommand{\arb}{\ensuremath{\mathrm{arb}}\xspace}
\newcommand{\eps}{\varepsilon}

\usepackage{xspace}
\newcommand{\NCC}{{\ensuremath{\textnormal{NCC}}}\xspace}

\newcommand{\NCCx}[1]{{\ensuremath{\textnormal{NCC}\textnormal{(}{#1}\textnormal{)}}}\xspace}
\newcommand{\NCCz}{{\ensuremath{\textnormal{NCC}_0}}\xspace}

\newcommand{\NCCzs}{{\ensuremath{\textnormal{NCC}^*_0}}\xspace}
\newcommand{\NCCzsx}[1]{{\ensuremath{\textnormal{NCC}^*_0({#1})}}\xspace}
\newcommand{\NCCzx}[1]{{\ensuremath{\textnormal{NCC}_0({#1})}}\xspace}
\renewcommand{\MPC}{{\ensuremath{\textnormal{MPC}}}\xspace}
\newcommand{\MPCx}[1]{{\ensuremath{\textnormal{MPC}}\textnormal{(}#1\textnormal{)}}\xspace}
\newcommand{\LOCAL}{{\ensuremath{\textnormal{LOCAL}}}\xspace}
\newcommand{\CONGEST}{{\ensuremath{\textnormal{CONGEST}}}\xspace}
\newcommand{\CC}{{\ensuremath{\textnormal{CC}}}\xspace}

\newcommand{\NClin}{{\ensuremath{\NC_{\textnormal{lin}}}}\xspace}

\newcommand{\NCx}[1]{{\ensuremath{\NC^{#1}}}\xspace}
\newcommand{\NClinx}[1]{{\ensuremath{\NC^{#1}_{\textnormal{lin}}}}\xspace}

\newcommand{\NCy}[1]{{\ensuremath{\NC_{\textnormal{#1}}}}\xspace}
\newcommand{\Pclass}{{\ensuremath{\P}}\xspace}
\newcommand{\Lclass}{{\ensuremath{\ComplexityFont{L}}}\xspace}

\title{Simulations between Massively Parallel Computing\\and Distributed Computing}
\author{Philipp Schneider\\
  \small CISPA Helmholtz Center for Information Security, Saarbrücken, Germany\footnote{Current affiliation}\\
  \small University of Bern, Bern, Switzerland\footnote{Affiliation where part of the work was done}\\
  \small\href{mailto:philipp.schneider@cispa.de}{\texttt{philipp.schneider@cispa.de}}\\
  \small\href{https://orcid.org/0000-0001-9660-1270}{ORCID 0000-0001-9660-1270}\\\\  
  Julian Werthmann\\
  \small Friedrich Schiller University Jena, Jena, Germany\footnotemark[1]\\
  \small Paderborn University, Paderborn, Germany\footnotemark[2]\\
  \small\href{mailto:jwerth@mail.upb.de}{\texttt{julian.werthmann@uni-jena.de}}\\
  \small\href{https://orcid.org/0000-0002-5110-5625}{ORCID 0000-0002-5110-5625}}
\date{}

\hypersetup{
  pdftitle={Simulations between Massively Parallel Computing and Distributed Computing},
  pdfauthor={Philipp Schneider and Julian Werthmann}
}
\begin{document}

\maketitle

\begin{abstract}
    We study how the Massively Parallel Computation ($\MPC$) model in the strongly sublinear regime relates to the Node-Capacitated Clique ($\NCC$) model, a bandwidth-parametrized generalization of the Congested Clique. In $\MPC$, $M$ machines with per-machine memory $S$ hold a partition of the input graph. In $\NCC$, we are given $n$ nodes that are themselves machines that can exchange a total of $C$ words per round with any other machine. We are interested in the strongly sublinear regime where $S=n^\delta$, for some constant $0 < \delta <1$ and $C = MS/n$, for which no prior simulation results are known.

    We explore when deterministic round-preserving simulations between these models are possible and when they are provably not, for different model parameters, problem families and graph classes.
    On the positive side, we provide techniques that allow, under certain restrictions, simulations with only constant overhead.
    On the negative side, we prove simulation impossibility results, which show that the limitations of our simulation results are inherent.
\end{abstract}

\clearpage
\tableofcontents

\newpage

\section{Introduction}

Distributed computing models differ in the resources they constrain, such as communication rounds, bandwidth, and memory. Classical graph-centric models such as \LOCAL and \CONGEST \cite{DBLP:journals/siamcomp/Linial92, Peleg2000Locality} identify computational entities with the nodes of the input graph and restrict communication to its edges. They measure complexity in communication rounds and treat local computation as free, with \CONGEST additionally restricting message size. These assumptions help isolate communication costs, but limit the portability of algorithms to settings where local computation matters and the communication infrastructure is decoupled from the input graph.

In contrast, the massively parallel computation (\MPC) model \cite{DBLP:conf/soda/KarloffSV10, DBLP:conf/pods/BeameKS13}, inspired by Google's MapReduce \cite{DBLP:journals/cacm/DeanG08}, decouples communication and input distribution from the input graph.
In $\MPC(M,S)$, the input is initially partitioned among $M$ machines with $S$ words of memory each, where the total memory $MS$ suffices to hold the input.
In each round, a machine may exchange at most $S$ words with arbitrary other machines.
The challenge is to organize information exchange under these memory and bandwidth constraints, for example when computing connected components or counting triangles.
For input graphs with $n$ nodes, much recent work focuses on the \textit{strongly sublinear regime} $S=n^\delta$ for a constant $0<\delta<1$, rather than the \textit{linear regime} $S=\Theta(n)$.
This regime is motivated by the need to process massive data sets with limited per-machine memory.

The congested clique (\CC) model \cite{DBLP:journals/siamcomp/LotkerPPP05} combines the any-to-any communication of \MPC with a graph-centric input distribution.
Its computational entities are the $n$ nodes of the input graph, each initially knowing all its incident edges.
Every pair of nodes may exchange one word per round, so communication is constrained by bandwidth rather than locality.
Algorithmic ideas are sometimes interchangeable between \CC and the \textit{linear regime} of \MPC, yielding similar results (e.g., \cite{DBLP:conf/podc/DoryFKL21, DBLP:conf/podc/ChangFGUZ19}).

We seek to formalize these connections through round-preserving simulations, which transfer algorithms between models with at most constant-factor overhead in round complexity.
Prior work establishes such connections between \CC and \textit{linear} \MPC under restrictions on local memory and local computation of \CC nodes, or the number of machines in \MPC \cite{DBLP:journals/tcs/HegemanP15,DBLP:journals/corr/abs-1802-10297}.
To the best of our knowledge, no general simulation or impossibility result is known without such restrictions, and even less is known in the \textit{strongly sublinear} regime.

Our main focus is to bridge this knowledge gap and provide simulation and impossibility results for the \MPC model, in particular for the highly bandwidth- and memory-restricted \textit{strongly sublinear} regime. Our model of comparison is a generalization of the \CC called the \textit{node capacitated clique} or $\NCC(C)$ model \cite{DBLP:conf/spaa/AugustineGGHSKL19}, where each node has a bandwidth of $C$ words per round (send \textit{and} receive) that it may exchange with any other node. The $\NCC(C)$ model allows studying distributed graph algorithms with limited bandwidth, in particular the \textit{strongly sublinear} regime $C=n^\delta$ for $0<\delta<1$ or the \textit{logarithmic regime} $C=\polylog(n)$. In all other modeling aspects the \NCC equals the \CC model, in particular, $\NCC(n)$ is equivalent to the \CC (implied by \cite{DBLP:conf/podc/Lenzen13}).

For $MS=nC$, $\MPC(M,S)$ and $\NCC(C)$ look relatively similar as they have the same total bandwidth and both allow communication independently of the input graph's topology.
They differ, however, in their memory constraints, input representation, and available parallelism.
We study how these differences affect the feasibility of simulations across the two models and characterize, as a function of model parameters, problem families, and graph structure, when such simulations exist and when they provably cannot exist.
The \NCC serves as an intermediate model for understanding how these modeling choices affect the relative power of \MPC and graph-centric models. An up-front conclusion is that simulation questions between \MPC and \NCC intersect several core areas of complexity theory, such as time and memory complexity, circuit size and depth, and information entropy, and have noteworthy connections to long-standing open problems.

\subsection{Models} 
\label{sec:model}

We consider synchronous message-passing models where computation machines are identical Turing machines that use the same constant-size program. Time proceeds in synchronous rounds. In each round a machine first receives all messages sent to it in the previous round, then updates its local state, and finally sends new messages. We use $n$ as the main parameter, to specify the size of the problem input and the number of machines (typically, $n$ equals the number of vertices of the input graph). Each machine has a unique identifier (ID) and knows a polynomial upper bound on the ID range. Other than that, IDs are arbitrary. Bandwidth is measured in machine words of $\Theta(\log n)$ bits, enough to carry an ID or an aggregate of polynomial size values. Model-specific constraints are given in the following definitions.

\subparagraph*{Node-Capacitated Clique} 
In the $\NCCx{C}$ model \cite{DBLP:conf/spaa/AugustineGGHSKL19} there are $n$ machines with a globally known ID set. Machines are identified with the vertices of the input graph, therefore we refer to them as nodes. In every round, a node may send/receive at most $C$ machine words in total, to/from any node in the network (any-to-any). 
Each node knows which node IDs correspond to its neighbors in the input graph.
Nodes may perform arbitrary finite local computations on their data. However, it is desirable to keep local computations efficient, e.g., in $\poly(n)$.
The $\NCCzx{C}$ model \cite{DBLP:journals/tpds/AugustineCCPSS22} keeps the same communication parameters but each node can initially only message its neighbors. Whenever it learns a new ID via a message, it may thereafter communicate directly with that node. The \NCCz model is strictly weaker than \NCC. For instance, learning the ID of a distant node has a lower bound logarithmic in the diameter as each round the distance to the farthest ID a given node knows can be at most doubled.

\subparagraph*{Massively Parallel Computation}
In the $\MPCx{M,S}$ model \cite{DBLP:conf/soda/KarloffSV10}, there are $M \in \bigO\big(\poly(n)\big)$ machines, each with $S$ words of local memory and an ID in $\{0,\dots,M-1\}$.
The total memory $MS$ is lower-bounded by both the input and output sizes. Minimizing the required total memory is typically an optimization criterion for \MPC algorithms.
Machines communicate in synchronous rounds, each exchanging at most $S$ words per round. Local computations are performed by Turing machines using at most $S$ words of space. 
These space restrictions reflect the practical motivation of \MPC as exchanged and computed data must fit into machine memory. Prior simulation results between \MPC-like models and \CC make comparable modeling choices \cite{DBLP:journals/tcs/HegemanP15,DBLP:journals/corr/abs-1802-10297}.
In particular, this bounds each terminating {deterministic} local computation by $n^{\bigO(S)}$ steps, although a $\poly(n)$ time bound is preferable.

\subparagraph*{Randomization} The \NCC and \MPC models permit randomized algorithms, however in this work we explicitly consider deterministic simulations. Therefore we model nodes and machines as deterministic Turing machines and thus our results only concern deterministic algorithms.
However our simulation results can be extended to randomized algorithms by supplying random choices through a separate random tape, that is, a read-only tape of random bits. For every fixed choice of these tapes, our positive simulation results deterministically emulate the corresponding execution. Hence, generating the random tapes in the same way as in the simulated algorithm preserves its output distribution and failure probability. 

\subparagraph*{Bandwidth Violations} The two models make slightly different provisions when the inbound message capacity of $C$ or $S$ is exceeded. The $\NCCx{C}$\ model of \cite{DBLP:conf/spaa/AugustineGGHSKL19} discards an adversarially chosen subset of messages. By contrast, \cite{DBLP:conf/soda/KarloffSV10} only stipulates that in the $\MPCx{M,S}$ model, reshuffling memories must respect memory bounds, which we interpret in the sense that an algorithm that could cause such a violation is illegal, i.e., not a viable solution to the problem. Allowing resolving capacity violations, even adversarially, potentially allows stronger algorithms than outright rejecting such algorithms, thus the policy of the \MPC model is more restrictive. To provide a level playing field we will apply the more restrictive policy of the \MPC model to both models. That is, only  \NCC or \MPC algorithms that ensure that no bandwidth violations occur are considered viable solutions to a problem.

\subsection{Simulation Definition}
\label{sec:sim_def}

In this work we are mainly, but not solely, concerned with graph problems, due to their close ties to the $\NCC$ model. We begin by carefully defining graph problems, as seemingly small details often matter for simulations across models.

\begin{definition}[Graph Problem]\label{def:problem}
A graph problem $\Pi$ is a set of problem instances $(I,S_I)$. $I=(G,f_I)$ is the input containing a graph $G = (V,E)$ and a function $f_I$ that maps nodes or edges to a constant number of words (i.e., input labels, e.g., edge weights).
$S_I$ is a set of legal \emph{outputs} for $I$, i.e., the solution of the graph problem, and we require any legal output to be of size at most $n^c$ words for an arbitrary but fixed $c>0$. 
If $f_I$ is trivial (e.g., $f_I(x)=0$ for all $x$) for every instance of $\Pi$, and the set of legal outputs $S_I$ is independent from the identifiers assigned to nodes then we call $\Pi$ unlabeled.\footnote{The specific solution from $S_I$ that a correct algorithm computes may still depend on identifiers (but not the set $S_I$ itself), so the identifiers can be exploited, e.g., as initial solution (coloring) or as tie breakers.}
\end{definition}

The requirement on the output size covers the most common graph problems, for example a mapping of nodes or edges to the colors of a legal node or edge coloring, or the graph diameter, or a bit representing whether the graph contains a triangle or even listing all triangles. However problems with huge output, e.g., listing all instances of a specific, super-constant size subgraph, which can be super-polynomially many, are not included.

Next, we fix how instances and solutions are represented in different distributed models. In particular, the algorithm output on instance $I$ must uniquely specify a single solution from $S_I$, even though it may be split across machines. In subsequent parts we refrain from discussing input and output representations in-depth and instead follow the usual conventions, stating the expected output representation for a given problem at first use. We formalize input and output representations as follows.

\begin{definition}[Input / output representation]\label{def:representation}
For $(I,S_I) \in \Pi$ and a distributed model $\mathcal M$, an \emph{input representation}  assigns the input-data of $I$ to the machines of $\mathcal M$ subject to model specific rules. Analogously, an \emph{output representation} is to be interpreted as an assignment of parts of a solution $s \in S_I$ to the machines in a model specific way, e.g.\ each \NCC-node stores its own color of a proper graph coloring. Neither the input nor the output representation are unique in general. The total length of the output of all machines must uniquely encode some $s \in S_I$, must be decodable in $\bigO(|s|)$ (centralized) steps and have size $\bigO(|s|)$ for some $s \in S_I$, where we allow individual machines to output an empty symbol $\bot$, which does not count towards the size.
\end{definition}

In the \MPC model an input representation places the problem input on the machines, subject to a per-machine memory budget of $S$ words but is otherwise arbitrarily distributed. For instance in sorting or aggregation tasks, the sorting or aggregation values are on the machines with at most $S$ values per machine. In case of a graph problem, any machine holds at most $S$ edges together with labels for the edges and endpoint nodes as specified by $f_I$.
Besides the memory restriction, we cannot make any assumption about the assignment of inputs on machines, i.e., we have to assume the worst case assignment.
The output representation is problem specific, for instance, given a labeling problem, the output representation must store each output label on some machine (e.g., vertex-color pairs). Or some machine must output whether a decision problem is a yes-instance.

In the \NCC model the problem input is usually related to a specific graph problem. Each vertex initially resides on its own machine together with the IDs of its neighbors and input labels associated to it and its incident edges, as specified by $f_I$. 
The output is again problem specific, e.g., in a labeling problem, $v$'s machine holds the output label for $v$ or its incident edges. When an algorithm is simulated in another model, the placement of the emulated machines determines how its output is distributed, and converting this to another output representation is not always efficient. We discuss the resulting representations alongside our simulation results and show that they meet the requirements of our applications, either directly or after efficient postprocessing.

\begin{definition}[Algorithmic solution in model $\mathcal M$]\label{def:solve}
An algorithm $\mathcal A$ for model $\mathcal M$ \emph{solves} $\Pi$ if, for every instance $I$ and every possible input representation, it outputs a representation of some legal output in $S_I$.
\end{definition}

Our notion of simulation is solution-oriented. It preserves solvability and round complexity, but does not require the simulating algorithm to reproduce the  execution of the simulated algorithm.
 
\begin{definition}[Simulation]\label{def:simulation}
Let $\mathcal A$ solve $\Pi$ in $T_{\mathcal A}$ rounds in model $\mathcal M$.  We say that $\mathcal A$ can be \emph{simulated} in model $\mathcal M'$ with overhead $t > 0$ if there exists an algorithm $\mathcal A'$ for $\mathcal M'$ that solves $\Pi$ using at most $\bigO(t \cdot T_{\mathcal A} + 1)$ rounds in total.\footnote{To keep the $\bigO(\cdot)$-notation well-defined even with multiple parameters, we assume that all parameters are functions of a single input parameter (typically $n$).}
\end{definition}

\begin{definition}[Model dominance]\label{def:dominance}
Let $\mathcal P$ be a set of problems.
For two models $\mathcal M_1$ and $\mathcal M_2$ and a value $t$ we write $\mathcal M_1\le_{\mathcal P,t}\mathcal M_2$ if every algorithm that solves some problem $\Pi \in \mathcal P$ in $\mathcal M_1$ can be simulated in $\mathcal M_2$ with overhead at most $t$. We omit $\mathcal P$ if the simulation holds universally for any type of algorithm or if $\mathcal P$ is clear from the context. We omit $t$ when $t \in \bigO(1)$.
\end{definition}

Hence $\mathcal M_1\le\mathcal M_2$ implies that $\mathcal M_2$ is at least as powerful as $\mathcal M_1$: any problem in the considered class $\mathcal P$ that is solvable in $T$ rounds in $\mathcal M_1$ is also solvable in $O(T)$ rounds in $\mathcal M_2$. Conversely, to prove $\mathcal M_1 \nleq \mathcal M_2$ it suffices to give a problem in $\mathcal P$ that admits an $O(T)$-round algorithm in $\mathcal M_1$ but requires $\omega(T)$ rounds in $\mathcal M_2$, i.e., it is equivalent to proving a separation between $\mathcal M_1$ and $\mathcal M_2$.
Framing the relative power of models in terms of simulation is conceptually stronger than comparing round complexities of specific problems, in particular, $\mathcal M_1\le\mathcal M_2$ implies that any algorithmic solution for $\mathcal M_1$ induces one for $\mathcal M_2$.

\subsection{Overview and Contributions}
\label{sec:overview}

We summarize our results and the ideas behind their proofs. Table~\ref{tab:simulation_comparison} compares the simulation bounds and their assumptions. Sections~\ref{sec:sim_mpc_in_ncc}--\ref{sec:nccz_in_mpc} give the full statements and proofs.

\newcolumntype{L}[1]{>{\footnotesize\raggedright\arraybackslash}m{#1}}
\newcolumntype{C}[1]{>{\footnotesize\centering\arraybackslash}m{#1}}
\makeatletter
\makeatletter

\begin{table}[t]
\centering
\setlength{\tabcolsep}{2pt}
\begin{tabular}{@{}L{.095\textwidth} C{.145\textwidth} C{.06\textwidth} C{.05\textwidth} C{.06\textwidth} C{.09\textwidth} C{.225\textwidth} C{.205\textwidth}@{}}
\toprule
Source
& Result
& $M$
& $S$
& $C$
& Overhead
& Requirements
& Notes \\
\midrule
\cite{DBLP:journals/corr/abs-1802-10297}
& $\CC \equiv \MPC$
& $n$
& $n$
& $n$
& $\bigO(1)$
& $\bigO(n)$ memory per \CC node
& direct result via routing~\cite{DBLP:conf/podc/Lenzen13} \\
\cite{DBLP:journals/tcs/HegemanP15}
& $\CC \leq \MPC$
& $\dfrac B n$
& $n$
& $n$
& $\bigO(1)$
& $\bigO(n)$ memory \& $\poly(n)$ work per \CC node
& total \CC memory and communication $\bigO(B)$ \\
\midrule
Thms.\newline \ref{thm:sim_mpc_in_ncc_impossible1}, \ref{thm:sim_mpc_in_ncc_impossible2}
& $\MPC \not\leq \NCC$
& any
& any
& $\dfrac{MS}{n}$
& $\omega(1)$
& $MS \leq \frac{a n}{w}$ for $w \in \omega(\log n)$
& triangle listing problem\\
Thm.~\ref{thm:ncc_sim_mpc}
& $\MPC \leq \NCC$
& $\geq n$
& any
& $\dfrac{MS}{n}$
& $\bigO(1)$
& \hspace{3.5mm}$MS \geq a n^{1+\delta}$,\newline $\delta \in (0,1)$
& output-compatible graph problems \\
Thm.~\ref{thm:sim_mpc_in_ncc_impossible3}
& $\MPC \not\leq \NCC$ decision prob.\
& any
& any
& $\dfrac{MS}{n^{1+\varepsilon}}$
& $\polylog(n)$
& \hspace{3mm} $MS \geq n^{1+\varepsilon}$,\newline \Pclass-bounded
& unless $\NC_{\mathrm{near\text{-}lin}}=\NC$ \\
\midrule
Thm.~\ref{thm:sim_nccz_in_mpc_infeasible}
& $\NCC_0 \not\leq \MPC$
& any
& any
& $1$
& no sim.
& $MS \leq (1+\varepsilon)^n$
& $\NCC_0 \leq_t \MPC$ impossible for any $t$\\
Thm.~\ref{thm:sim_nccz_in_mpc_inefficient} & $\NCC_0 \not\leq \MPC$ & $n^{\bigO(1)}$ & $n^\delta$ & $1$ & $\Omega(n/S)$ & $\poly(n)$ random oracle queries, $\delta\!\in\!(0,1)$ & works with polynomial local work in $\NCCz$\\
Thm.~\ref{thm:sim_ncczs_in_mpc}, (Cor.~\ref{cor:sim_ncczs_in_mpc})
& $\NCCzs \leq \MPC$
& $n$
& $n^{4\delta}$
& $n^\delta$
& $\bigO(1)$
& unlabeled graphs, $a\leq n^\delta$, $\delta\in(0,1/4)$
& see Definition~\ref{def:ncczs} for \NCCzs\\
\bottomrule
\end{tabular}
\caption{Comparison of simulation results between $\MPC(M,S)$ and $\NCC(C)$ and variations. See Sections \ref{sec:model} and \ref{sec:sim_def} for definitions and notation. We view $\CC$ as $\NCC(n)$. Entries in the columns $M$, $S$, and $C$ are stated up to constant factors. We define $a$ as arboricity of the input graph $G$.}
\label{tab:simulation_comparison}
\end{table}

\subparagraph*{Impossibility of \MPC Simulation in \NCC}

We first show that matching total bandwidth, $MS=nC$, does not suffice for a constant-overhead simulation of $\MPC(M,S)$ in $\NCC(C)$.
For this we consider the triangle listing problem. In $\MPCx{M,S}$ this problem can be solved in constant time when given enough machines.
We then show that this parallelism cannot be replicated in $\NCCx{C}$, even when $nC = MS$.
Our hard instances are ``lollipop'' graphs: a long path with a dense random subgraph on $a$ nodes at one end.
Their arboricity\footnote{The \emph{arboricity} of $G$ is the minimum number of forests into which the edges of $G$ can be partitioned.} is at most $a$ and is $\Theta(a)$ in expectation, so the construction parameter $a$ also serves as an arboricity bound.
The path keeps the graph sparse and connected, while the dense part contains $\Omega(a^2)$ random edges and $\Omega(a^3)$ triangles in expectation.

In $\MPCx{M,S}$, the candidate triangles in this dense part can be checked in parallel by many machines.

In $\NCCx{C}$, however, the information about these triangles is initially trapped at the relatively few nodes of the dense part.
Since every triangle must be listed somewhere, either many triangles are listed by nodes on the path, in which case the dense part has to reveal $\Omega(a^2)$ bits of random edge information in total, corresponding to $\Omega(a)$ bits per dense node on average, or many triangles are listed inside the dense part, in which case some dense node has to learn $\Omega(a^{4/3})$ such bits.
As each node can communicate only $O(C\log n)$ bits per round, this yields a lower bound of order $\Omega(a/(C\log n))$ rounds.
For the large-arboricity case, a single lollipop instance suffices.
For smaller arboricity values, we instead use many independent lollipop copies joined together, so that the same information bottleneck remains while the overall graph has the desired size and arboricity range.
This yields the round lower bounds given in Theorems~\ref{thm:sim_mpc_in_ncc_impossible1} and~\ref{thm:sim_mpc_in_ncc_impossible2}, which we summarize in the subsequent theorem.

\begin{theorem*}[cf.\ Thms.\ \ref{thm:sim_mpc_in_ncc_impossible1}, \ref{thm:sim_mpc_in_ncc_impossible2}]
    There is a graph problem $\Pi$ that, on graphs of arboricity at most $a$, can be solved in constant rounds in $\MPCx{M,S}$ using total memory $MS\leq \tfrac{a\cdot n}{w}$ for any
    $w\in\omega(\log n)$, but requires $\omega(1)$ rounds in $\NCC\big(\tfrac{MS}{n}\big)$.
\end{theorem*}

\subparagraph*{\MPC Simulation in \NCC} We then explore to which extent algorithms for the $\MPC$ model can be simulated in $\NCC$ in Section \ref{sec:sim_mpc_in_ncc}. When $M\geq n$ and each node initially holds at most $MS/n$ input words, the $\NCC(MS/n)$ nodes can emulate each round of an $\MPC(M,S)$ algorithm in $\bigO(1)$ rounds.
This allows us to use constant-time sorting and aggregation in strongly sublinear $\NCC$ by simulating $\MPC$ techniques (cf.\ \Cref{sec:mpc_routines}). In particular, this simulation extends to graphs whose maximum degree is within a constant factor of the average degree.
We then turn to graphs where nodes can have arbitrarily large degree. Here, the goal is to reassign edges so that the number of assigned edges is closer to the average, allowing us to fall back on the simulation for well-distributed inputs.

The core obstacle for this edge reassignment is that high degree nodes can only directly communicate the fact that they have high degree to a small fraction of their neighbors in constant time with the limited available bandwidth. To enable this reassignment we need to consider the identification problem \cite{DBLP:conf/spaa/AugustineGGHSKL19}. 
 Nodes are designated as \emph{learning} or \emph{playing} (or neither), and each learning node must learn all of its playing neighbors.
We give (in Theorem~\ref{thm:identification}) a constant-round, deterministic algorithm for the identification problem in the strongly sublinear \NCC regime that exploits simulated \MPC subroutines on evenly distributed inputs.

The algorithm lets each learning node identify its non-playing neighbors through a sequence of shrinking ID intervals. In each iteration, intervals that may contain a non-playing neighbor are divided into smaller intervals, and simulated \MPC aggregation determines which of them remain relevant. Each refinement reduces the interval length by a polynomial factor, so after constantly many iterations the remaining intervals are singletons and the learning node can infer all of its playing neighbors.

We then show how to use node-identification to compute an edge orientation where the out-degree is bounded by the arboricity $a$ times a small factor.
This reassigns every edge to one endpoint and brings the input close to the well-spread setting needed for the \MPC simulation.

Finally, we show that the simulated outputs can be redistributed in constantly many additional rounds to satisfy the $\NCC$ output conventions required by our applications.

Together, these ingredients let $\NCC(MS/n)$ emulate the rounds of an $\MPC(M,S)$ algorithm once the total memory is sufficiently large relative to the arboricity. We paraphrase Theorem~\ref{thm:ncc_sim_mpc}, which gives the dependence of the additive overhead on $\delta$.

\begin{theorem*}[cf.\ Thm.\ \ref{thm:ncc_sim_mpc}]
    Let $M\geq n$ and fix a constant $\delta\in(0,1)$. If $MS\geq a\cdot n^{1+\delta}$, then any $\MPCx{M,S}$ algorithm that solves an output-compatible graph problem (Definition~\ref{def:output_compatibility}) on graphs of arboricity at most $a$ can be simulated in $\NCC\big(\tfrac{MS}{n}\big)$ with constant \emph{additive} overhead.
\end{theorem*}

Compared to our negative simulation result from above the remaining gap to the negative result is only a factor of $\bigO(n^\delta \log n)$ for arbitrarily small constant $\delta >0$.
The condition $MS \geq a \cdot n^{1+\delta}$ is also close to the memory needed to store the input itself, as an $n$-node graph with arboricity $a$ has at least $\Omega(n+a^2)$ and up to $\bigO(an)$ nodes and edges (see Lemma \ref{lem:arb_min_edges}). The theorem transfers a range of known strongly sublinear \MPC algorithms to \NCC, the resulting \NCC algorithms remain strongly sublinear whenever $MS \leq n^{2-\eps}$ for some constant $\eps>0$. We list a selection below.

\begin{itemize}
    \item Maximal independent sets and maximal matchings on graphs with bounded arboricity $a = \polylog (n)$ can be computed in $\bigO(\log\log n)$ rounds in $\NCC(n^\delta)$, due to \cite{DBLP:conf/wdag/GhaffariGJ20}.
    \item $(\Delta + 1)$-colorings on graphs with bounded arboricity $a \leq n^{\delta}$ can be computed in $\bigO(\log\log\log n)$ rounds in the $\NCC(n^{2\delta})$ model, implied by \cite{DBLP:conf/podc/CzumajDP21}.
    \item Triangle listing, i.e., the problem of outputting all triangles in a graph with arboricity $a \leq n^\delta$, can be solved in constant time in the $\NCC(n^{2\delta})$ model, implied by the constant round $\MPC(M,n^\delta)$ algorithm with total space $MS \in \bigO(a|E|) \subseteq \bigO(a^2n)$ by \cite{DBLP:conf/podc/Liu024}.
\end{itemize}

\subparagraph*{Impossibility of Simulating \MPC Decision Problems in \NCC} The triangle-listing lower bounds exploit large outputs to create an information bottleneck.
We therefore ask whether graph decision problems permit simulations with substantially less total bandwidth. 
Unfortunately, even for decision problems with computationally limited machines and nodes, such a significantly better simulation would have surprising consequences in circuit complexity.
Concretely, it would imply that every decision problem solvable by polynomial-size, polylogarithmic-depth circuits can also be solved by near-linear-size, polylogarithmic-depth circuits. In short, it would imply $\NCy{near\text{-}lin}=\NC$, which is a major open question.
We refer to Section~\ref{sec:sim_mpc_in_ncc_dec_problems_impossible} for a refresher on circuit complexity notation and definitions.
The technical bridge to this consequence is Lemma~\ref{lem:ncc_to_circuit}.
It shows that a $\Pclass$-bounded $\NCCx{C}$ algorithm\footnote{A $\Pclass$-bounded $\NCCx{C}$ algorithm can only perform polynomial computations on its local information each round.} with polylogarithmic local input, bandwidth and round complexity can be transformed into an $\Lclass$-uniform circuit\footnote{Let $C_n$ be the family of circuits that decides decision problems of $n$ bit inputs. $\Lclass$-uniformity requires that $C_n$ is computable by a Turing machine using only $\bigO(\log n)$ additional space on input $1^n$. Uniformity ensures $\NC \subseteq \Pclass$.} whose size is near-linear in the number of nodes.

The subtle point in this transformation is circuit uniformity: the wiring of an $\Lclass$-uniform circuit depends only on the input length, whereas the communication pattern of an \NCC algorithm may depend on the actual input graph. We therefore first make the \NCC algorithm input-oblivious, so that the possible communicating pairs in each round depend only on the instance size. Using oblivious sorting and routing through simulated \MPC subroutines, this incurs only polylogarithmic overhead (see Lemma \ref{lem:mpc_oblivious} and Corollary \ref{cor:ncc_oblivious}). Each node-round computation can then be represented by a small circuit and these circuits can be wired together according to a log-space computable pattern.

We then use this circuit transformation to prove the conditional impossibility result. Suppose that every polynomially bounded $\MPC(M,S)$ decision algorithm could be simulated in $\NCC(C)$ with polylogarithmic overhead even when $C = MS/n^{1+\eps}$ for some constant $\eps>0$. Take an arbitrary problem $\Pi \in \NC$, represented by a polynomial-size, polylogarithmic-depth uniform circuit family. We encode an input string as a sparse graph and construct an $\MPC$ algorithm that evaluates the circuit layer by layer. By the assumed simulation, this $\MPC$ algorithm can be converted into a polylogarithmic-round \NCC algorithm with small bandwidth. Applying the transformation above (Lemma~\ref{lem:ncc_to_circuit}) then turns this \NCC algorithm back into a uniform circuit of near-linear size in the graph size and polylogarithmic depth. With the chosen graph encoding, this shrinks the original circuit by a polynomial factor while preserving polylogarithmic depth. Repeating the argument constantly many times compresses any $\NC$ circuit to near-linear size, yielding $\NC_{\mathrm{near\text{-}lin}}=\NC$. Since this is a major open question, such a simulation is unlikely. The following theorem summarizes the result; see Theorem~\ref{thm:sim_mpc_in_ncc_impossible3} for details.

\begin{theorem*}[cf.\ Thm.\ \ref{thm:sim_mpc_in_ncc_impossible3}]
    Assume that machines and nodes are $\Pclass$-bounded. For any constant $\varepsilon>0$ and any fixed $t=\polylog(n)$, there is a graph decision problem $\Pi$ that has an $\MPCx{M,S}$ algorithm with $MS\ge n^{1+\varepsilon}$ that solves it in $r$ rounds, but any $\NCCx{MS/n^{1+\varepsilon}}$ algorithm that solves $\Pi$ takes more than $r \cdot t$ rounds, unless $\NCy{near\text{-}lin}=\NC$.
\end{theorem*}

Could the preceding separation instead be proved unconditionally? We show that this would itself constitute major progress in circuit complexity. Assume that $\NCy{near\text{-}lin}=\NC$ and consider a $\Pclass$-bounded $\MPC(M,S)$ algorithm $\mathcal A$ for a graph decision problem, where $S=\polylog(n)$ but $M$ may be a large polynomial in $n$. Using the analogous input-oblivious transformation and circuit representation for \MPC (Lemmas \ref{lem:mpc_oblivious} and \ref{lem:mpc_to_circuit}), $\mathcal A$ can be transformed into an $\Lclass$-uniform circuit family of polylogarithmic depth.

Under the assumption $\NCy{near\text{-}lin}=\NC$, this circuit can be replaced by an equivalent near-linear size, polylogarithmic-depth circuit.
For a graph of maximum degree $\Delta$, the input has size $m=\bigO(n\Delta\log n)$, so near-linear size means $\bigO(n\Delta\polylog \, n)$ gates.
Such a circuit can then be simulated directly in $\NCC(\Delta)$ by assigning the gates of each layer canonically to the $n$ nodes, such that each node evaluates only $\bigO(\Delta)$ gates per simulated layer (stretch large layers by a polylogarithmic factor, if needed).
Since every gate has constant fan-in, each node needs to receive only $\bigO(\Delta)$ input values per simulated layer, which fits the $\NCC(\Delta)$ bandwidth.
Thus, assuming $\NCy{near\text{-}lin}=\NC$, the original \MPC algorithm $\mathcal A$ can in fact be simulated in $\NCC(\Delta)$ with only polylogarithmic overhead.
Consequently, an unconditional proof that such a simulation is impossible would imply $\NCy{near\text{-}lin}\neq\NC$.

\begin{theorem*}[cf.\ Thm.\ \ref{thm:sim_mpc_in_ncc_impossible4}]
Assume that machines and nodes are $\Pclass$-bounded. Let $\mathcal A$ be a $\polylog(n)$-round $\MPCx{M,S}$ algorithm with $S=\polylog(n)$ that solves a decision problem on labeled $n$-node graphs of maximum degree $\Delta$. An unconditional proof that $\mathcal A$ cannot be simulated in $\NCCx{\Delta}$ with any polylogarithmic overhead would imply $\NCy{near\text{-}lin}\neq\NC$.
\end{theorem*}

\subparagraph*{Impossibility of \NCCz Simulation in \MPC} 
Although the \MPC model has an edge over \NCC for problems where parallelism and an even input distribution help, the per-machine memory bound makes simulating $\NCC$ algorithms challenging.
For this reason, in Section~\ref{sec:nccz_in_mpc_impossibility}, we first identify conditions that prohibit the simulation of \NCC algorithms in the \MPC model. 

We start by showing that even algorithms for the weaker $\NCCz$ model cannot be simulated in \MPC unless the total memory $MS$ is unreasonably large.
The construction is based on regular expressions with squaring, which use the typical regular operations with an additional squaring operation $r^2 := r\cdot r$. For a pair $(r_1,r_2)$ of regular expressions with squaring, the problem is to decide whether they define the same language $L(r_1)=L(r_2)$.
Regular expressions with squaring still define regular languages, but the squaring operator makes their representation exponentially more succinct and \cite{MeyerStockmeyer1972} showed that recognizing equivalence of two such expressions requires exponential space.

We encode an instance $(r_1,r_2)$ as a connected, unlabeled graph. Then the problem is solvable in $\NCCz(1)$ in a polynomial number of rounds by collecting the entire graph at a single node and solving the decoded decision problem locally.
However, if the same graph problem were solvable in $\MPC(M,S)$ with total memory $MS \le (1+\varepsilon)^n$ for arbitrarily small $\varepsilon>0$, then a Turing machine could simulate the \MPC computation using only $\bigO(MS\log n)$ bits of space, which would contradict the lower bound by \cite{MeyerStockmeyer1972}.
Hence any \MPC simulation of this general $\NCCz$ computation requires exponential total memory.
See Theorem~\ref{thm:sim_nccz_in_mpc_infeasible} for more details.

\begin{theorem*}[cf.\ Thm.\ \ref{thm:sim_nccz_in_mpc_infeasible}]
    There is a graph problem $\Pi$ and an $\varepsilon > 0$ such that $\Pi$ is solvable in the $\NCCzx{1}$ model but unsolvable in the $\MPCx{M,S}$ model for $MS \leq (1+\varepsilon)^n$.
\end{theorem*}

Even with arbitrarily large polynomial total space $MS$, there exists a problem that, in the random-oracle model, can be solved instantly, i.e., without any communication, in $\NCCz$ but takes much longer in sublinear-memory $\MPC$.
Both models have access to the same random oracle, and each node or machine may make polynomially many oracle queries per local computation phase.
This separates the models for a different reason than Theorem~\ref{thm:sim_nccz_in_mpc_infeasible}, namely the inability of individual $\MPC$ machines to locally hold and process all information concentrated at a single $\NCCz$ node.

The hard instance is essentially a labeled star.
The center node has $\Theta(n)$ incident edge labels of one word each, which are grouped into blocks of $\Theta(\log n)$ words.
These blocks define a sequential computation of a function $h$ using a construction adapted from~\cite{DBLP:conf/spaa/ChungHS20}.
Each step queries the random oracle on the concatenation of the next input block and the current state, and part of the answer becomes the next state.
In $\NCCz$, each node, in particular the center, initially knows all incident labels and can therefore compute the entire chain locally in $0$ communication rounds.

In $\MPC(M,S)$, however, a machine can store only $\bigO(S/\log n)$ such blocks.
To advance along the chain, it must combine the current state with the next block.
The random oracle makes later states difficult to predict before the preceding queries have been answered, and newly computed states can only be passed to another machine between rounds.
An adaptation of the compression argument in~\cite{DBLP:conf/spaa/ChungHS20} shows that, for uniformly random input blocks, useful progress along the chain is limited to $\bigO(S/\log n)$ steps per round with high probability, even when machines store arbitrary encodings and make speculative oracle queries.
Since the center's input consists of $\Theta(n/\log n)$ blocks, this gives an $\Omega(n/S)$ round lower bound for every polynomial number of machines.
For more details, see Theorem~\ref{thm:sim_nccz_in_mpc_inefficient}.

\begin{theorem*}[cf.\ Thm.\ \ref{thm:sim_nccz_in_mpc_inefficient}]
In the random-oracle model, there exists a labeled graph problem that is solvable in $0$ rounds in $\NCCz(1)$ but requires $\Omega(n/S)$ rounds in $\MPC{M,S}$ for every $S=n^\delta$ with fixed $0<\delta<1$ and every $M=n^{\bigO(1)}$.
The lower bound applies to algorithms making $\poly(n)$ oracle queries per machine in each local computation phase.
\end{theorem*}

\subparagraph*{\NCCzs Simulation in \MPC} We study simulating $\NCC$ algorithms in the $\MPC$ model in Section \ref{sec:nccz_in_mpc}. Unfortunately, the impossibility results from Section \ref{sec:nccz_in_mpc_impossibility} make us rather pessimistic about simulating even the weaker $\NCCz$ model in $\MPC$. The core obstacle is local concentration of information in \NCCz: whenever a single $\NCCz$ node holds a lot of information, a non-distributive or memory-intensive function $h$ on that information can be evaluated instantly, but the computation of $h$ in $\MPC$  would either require heavy inter-machine coordination or be outright infeasible.

To isolate what \emph{can} be simulated, we adopt a stricter model, the $\NCCzs$. Here, nodes may not keep substantial local information like neighbor IDs or weights of incident edges. Instead, a node $v$ knows only its own ID, port numbers of its initial neighbors (used to contact them) and a locally computable function $f_v$ that maps any learned ID to $\bot$ if it is not an initial neighbor, otherwise to a port number identifying that neighbor locally. Aside from $f_v$, each node has only $\bigO(C)$ additional local memory. This preserves the key $\NCCz$ feature of neighbor identification, while aligning memory restriction with $\MPC$. Under these constraints we give a simulation of $\NCCzs(C)$ in $\MPC$, see Theorem \ref{thm:sim_ncczs_in_mpc} and Corollary \ref{cor:sim_ncczs_in_mpc}.

\begin{theorem*}[cf.\ Thm.\ \ref{thm:sim_ncczs_in_mpc} or Cor.\ \ref{cor:sim_ncczs_in_mpc}]
    Any $\NCCzsx{n^\delta}$ algorithm on unlabeled graphs of arboricity at most $n^\delta$, for $\delta\in(0,\tfrac14)$, can be simulated in $\MPC(n,n^{4\delta})$ with constant overhead.
\end{theorem*}

For implementing the neighborhood function $f_v$ in the \MPC model, the main challenge is that a single machine cannot store the full set of neighbor identifiers of a high-degree node $v$.
Nevertheless, a simulated \NCCzs node must still be able to test whether an identifier learned during the simulation belongs to its initial neighborhood and, if so, map it to the corresponding local port.
We therefore compute a compressed and distributed representation of every neighborhood, so that the machine $M(v)$ simulating $v$ can evaluate $f_v$ via several helper machines.

The preprocessing uses an optimized forest decomposition.
It partitions the edges into oriented forests and splits these forests into two parts (Lemma \ref{lem:mpc_forest_decomposition}).
In one part, the forests have bounded in-degree, so the corresponding neighbor identifiers can be stored explicitly: each machine only has to remember a bounded number of such neighbors.
In the other part, the forests may have high in-degree, but the decomposition guarantees that they have bounded depth.
For these bounded-depth forests, we recompute temporary identifiers according to a breadth-first traversal.
This makes the children of every node appear in a contiguous interval of the new identifier space.
Hence a machine does not have to store all child identifiers explicitly; it only stores a compact description of the interval, such as the first position and the number of children.

The original machine identifier of a neighbor can then be recovered from this compressed representation by contacting the machine responsible for the corresponding position in the interval.
In this way, the full neighborhood information is not stored locally at $M(v)$, but $M(v)$ can still evaluate the neighborhood function $f_v$ on the identifiers that a simulated \NCCzs node may store locally.

\subsection{Further Related Work}
\subparagraph*{The \MPC and \NCC models: Background and Basic Routines} A precursor of the contemporary \MPC model was formalized by \cite{DBLP:conf/soda/KarloffSV10}
as a round-based map/shuffle/reduce abstraction of the MapReduce paradigm \cite{DBLP:journals/cacm/DeanG08} with subsequent refinements by \cite{DBLP:conf/isaac/GoodrichSZ11}. The \MPC model, as it is in use today, was introduced by \cite{DBLP:conf/pods/BeameKS13, DBLP:journals/jacm/BeameKS17}.
We utilize many of the basic routines for the strongly sublinear \MPC regime for instance for sorting, aggregation, prefix-sums \cite{ DBLP:conf/soda/KarloffSV10, DBLP:conf/isaac/GoodrichSZ11} and slight generalizations thereof. An overview is provided in Appendix \ref{sec:mpc_routines}.

The node capacitated clique model (\NCC) was conceptualized by \cite{DBLP:conf/spaa/AugustineGGHSKL19} as a more restrictive version of the congested clique \CC model \cite{DBLP:journals/siamcomp/LotkerPPP05}, where instead of allowing any pair of nodes to exchange a machine word ($\bigO(\log n)$ bits), each node can send and receive at most $\bigO(\log n)$ machine words per round. We use a generalized version $\NCC(C)$ with a capacity parameter $C$ that specifies the number of words a node may send or receive per round.

Since the algorithms considered by \cite{DBLP:conf/spaa/AugustineGGHSKL19} are designed for the \textit{logarithmic} regime $\NCC(\log n)$, they are conceptually different to our designs for the  \textit{strongly sublinear} regime $\NCC(n^\delta), \delta \in (0,1)$. In particular, \cite{DBLP:conf/spaa/AugustineGGHSKL19} presents \textit{randomized} solutions for various tasks (e.g., MIS, maximal matching, $\bigO(a)$-Coloring) that depend polylogarithmically on $n$ and linearly on the graph arboricity $a$. Compare this to the $\NCC(n^\delta)$-algorithms implied by our simulation results (see above), which are exponentially faster (MIS, maximal matching) or even double-exponentially faster (coloring) for certain arboricity ranges.
A notable subproblem considered by \cite{DBLP:conf/spaa/AugustineGGHSKL19} gives a randomized solution for the identification problem in $\bigO(\log n + a)$ rounds. We propose a deterministic algorithm for the strongly sublinear \NCC (as subroutine for the simulation) that takes $\bigO(1)$ rounds for certain parameter ranges.

\subparagraph*{\CONGEST and \CC Simulations}
\cite{DBLP:conf/podc/GhaffariKS17} develops randomized routing techniques for \CONGEST on well-connected graphs, which in particular yield a simulation of \CC communication in \CONGEST. Their simulation of a $t$-round \CC algorithm can be expressed in terms of the mixing time as taking $\widetilde \bigO\big(t(n\tau_{\mathrm{mix}}+\Delta^2\tau_{\mathrm{mix}}^2)\big)$ \CONGEST rounds w.h.p., where $\Delta$ is the maximum degree and \smash{$\widetilde \bigO$} suppresses $\log n$ terms. \cite{DBLP:conf/spaa/Censor-HillelLP21} improves the dependence on graph density, giving a simulation in $t(n^2/m)\tau_{\mathrm{mix}}\cdot n^{o(1)}$ rounds w.h.p., subject to a mild bound on the input and output size at each node. Thus, up to subpolynomial factors, the latter replaces the roughly linear dependence on $n$ by $n^2/m$, which is roughly $n$ divided by the average degree, while also avoiding the dependence on $\Delta$. These results concern simulations between graph-centric distributed models with one processor per input node. Our setting is qualitatively different as we compare the bandwidth-limited \NCC model with \MPC, which may exploit polynomially more machines and whose input distribution is decoupled from the input graph.

\subparagraph*{\MPC Simulations}
\cite{DBLP:journals/tcs/HegemanP15} shows that a \CC algorithm with per-node memory of $\bigO(n)$ and polynomial local work per round can be efficiently simulated in \textit{linear} \MPC with constant-factor round overhead, using as many machines as needed to match the simulated \CC routine’s total used memory and bandwidth (which may be fewer than $n$ machines). \cite{DBLP:journals/corr/abs-1802-10297} shows that the \textit{linear} regime \MPC with $n$ machines is essentially equivalent to \CC with $ \bigO(n)$ per-node memory up to applications of the \CC routing algorithm by \cite{DBLP:conf/podc/Lenzen13}.

\cite{DBLP:journals/corr/abs-1802-10297} also gives a simulation of $T$-round \CONGEST algorithms in linear \MPC with $\bigO(T \, |E|/n)$ machines if the simulated \CONGEST algorithm is severely memory restricted (as low as $\polylog(n)$ bits). \cite{DBLP:conf/fsttcs/KothapalliPP20} simulates state-congested, $\alpha$-sparse \CONGEST routines in strongly sublinear \MPC. However, directly carrying their result over to \CC requires restricting per-node expected bandwidth to at most $n^{o(1)}$ and per-node memory to $\bigO(n \, \polylog (n))$, in order to stay within the strongly sublinear \MPC. 
\cite{DBLP:conf/soda/KarloffSV10}  proved that any $t$-step CREW PRAM using $O(n^{2-2\varepsilon})$ processors can be executed in $O(t)$ MapReduce rounds \cite{DBLP:conf/soda/KarloffSV10}, if the MapReduce memory matches the shared memory used by the PRAM algorithm.
\cite{DBLP:conf/isaac/GoodrichSZ11} strengthened this by simulating CRCW PRAM algorithms with commutative write conflict resolution (e.g., min/sum/max).

The results above do not address general simulations of $\MPC(M,S)$ algorithms in \CC beyond essentially linear settings. We instead parameterize the per-node bandwidth through $\NCC(C)$ and study simulations also in the strongly sublinear \MPC regime. Conversely, prior simulations of \CC algorithms in \MPC impose restrictions on local memory or computation of \CC nodes; our impossibility results for simulating \NCC (and hence \CC) in \MPC show why such restrictions are necessary (Section \ref{sec:nccz_in_mpc_impossibility}).

\subparagraph*{Lower Bounds for Strongly Sublinear \MPC and Relations to Hardness Conjectures} In \cite{DBLP:journals/jacm/RoughgardenVW18} it is shown that any output bit of a $r$-round \MPC algorithm with a receive bound of $s$ bits per round, must correspond to a polynomial representation of a boolean function whose degree is bounded by $s^r$. This means that a degree lower bound for a boolean function implies a round-lower bound for \MPC. They further show that, given $\poly(n)$ machines, one can compute $k = \lfloor \log_2 n \rfloor$ layers of an $\NC$ circuit in just a single round of strongly sublinear \MPC, by concentrating all inputs required to compute a gate $k$ levels up the circuit, on a single machine (see Section \ref{sec:sim_mpc_in_ncc_dec_problems_impossible} for a refresher on circuit complexity). This ultimately implies that a super-constant lower bound for the strongly sublinear \MPC (with $\poly(n)$ machines) for a problem in \Pclass would necessitate a super-logarithmic depth circuit, and thus separate \Pclass from $\NC^1$, which demonstrates the hardness of obtaining such a lower bound. \cite{DBLP:conf/isaac/FreiW19} improves the circuit simulation of a depth-$d$ circuit in $O(d/\log n)$ \MPC rounds, by using fewer machines, specifically the number of machines is strongly sublinear in the size of the simulated circuit.

Acknowledging the hardness of achieving unconditional super-constant lower bounds for strongly sublinear \MPC, \cite{DBLP:conf/focs/GhaffariKU19} instead turns to lower bounds conditioned on the 1v2 cycles conjecture. The 1v2-cycles conjecture stipulates that, even when promised that an input graph forms either an $n$-cycle or two $n/2$-cycles, it takes $\Omega(\log n)$ rounds
in the strongly sublinear \MPC model to distinguish the two cases. 
\cite{DBLP:conf/focs/GhaffariKU19} lifts \LOCAL lower bounds to strongly sublinear \MPC for \textit{component-stable} algorithms\footnote{An \MPC algorithm is called \emph{component-stable} if, whenever the input graph is a disjoint union of components, the algorithm’s output on any component depends only on that component.}, showing that any $o(\log D)$-round \MPC algorithm for a problem with \LOCAL \textit{lower bound} $D$ would yield an $o(\log n)$ algorithm for 1v2 cycles. 
Assuming the conjectured hardness for 1v2-cycles, this yields conditional super-constant $\MPC$ lower bounds for several locally checkable labeling problems (LCLs) and related problems, namely, $\Omega(\log\log n)$ for MIS, constant approximation matching and vertex cover, $\Omega(\sqrt{\log\log n})$ for $(\Delta \!+\! 1)$-coloring; and $\Omega(\log\log\log n)$ for LLL and sinkless orientation.

\cite{DBLP:journals/dc/NanongkaiS22} shows that the 1v2-cycles conjecture is equivalent to $\Lclass \not\subseteq \ComplexityFont{SL-MPC}$, where the former is the set of problems solvable with logarithmic space and the latter describes the class of problems solvable in sublogarithmic rounds in the strongly sublinear MPC model. \cite{DBLP:journals/dc/NanongkaiS22} summarizes hardness results for several problems conditioned on $\Lclass \not\subseteq \ComplexityFont{SL-MPC}$, thus providing a situation where the repercussions of refuting $\Lclass \not\subseteq \ComplexityFont{SL-MPC}$ are similar to refuting $\P \neq \NP$ in the sense that a large number of so far intractable problems would suddenly be solvable within the specified bounds. They also show that some problems can be conditioned on $\ComplexityFont{NL} \not\subseteq \ComplexityFont{SL-MPC}$, (which is even harder to refute).

\subparagraph*{Intractability of Obtaining Small Circuits for Problems in \NC}
Our simulation impossibility results do not rely on the 1v2-cycles conjecture (we ensure that lower-bound graphs are connected). Instead, our conditional lower bound uses the hypothesis $\NC\ne\NCy{near-lin}$, which is closely related to the question $\NC\stackrel{?}{=}\NClin$. $\NC$ is the class of decision problems, i.e., Boolean functions, that admit uniform circuits with polynomial size, constant fan-in, and polylogarithmic depth. The subclass $\NClin$ requires linear-size circuits while keeping polylogarithmic depth (see Section \ref{sec:sim_mpc_in_ncc_dec_problems_impossible} for details).

Answering the question $\NC \stackrel{?}{=} \NClin$ is notoriously hard. On one hand, no problems in $\NC$ are known that provably have a $\omega(n)$ size lower bound even without a  depth-bound \cite{DBLP:conf/innovations/GolovnevKW21}. On the other hand, 
there is no known method that shows how to reduce polynomial-size polylogarithmic-depth circuits to size $\bigO(n)$ while keeping the depth polylogarithmic. We consider the intermediate class $\NCy{near-lin}$ of near-linear size $n \cdot \polylog(n)$, with $\NClin\subseteq\NCy{near-lin}\subseteq\NC$. The question \smash{$\NC\stackrel{?}{=}\NCy{near-lin}$} (A) is similarly hard as \smash{$\NC\stackrel{?}{=}\NClin$} (B). A negative answer for (A) implies a negative answer for (B), while a positive answer for (A) would constitute major progress towards (B).

\section{Almost Tight \MPC Simulations in \NCC}
\label{sec:sim_mpc_in_ncc}

We start our comparison of the relative power of the \MPC and \NCC model families by investigating simulations of \MPC algorithms in \NCC. Given that we have enough bandwidth in the \NCC model to match total memory in the \MPC model, i.e., $nC=MS$ one could get the impression that the former is strictly stronger due to its unbounded local memory. However, this is not the case.
We first establish an information-theoretic impossibility result for graphs of sufficiently large arboricity using a problem with a large output (see Appendix \ref{sec:graph_properties} for the definition of arboricity).
We then give constant-overhead simulations under an arboricity condition that nearly matches the impossibility bound.

\subsection{Impossibility of Simulating \MPC in \NCC}
\label{sec:sim_mpc_in_ncc_impossible}

Our separating problem is \emph{triangle listing}, where every triangle must be reported by at least one \NCC node or \MPC machine.
This is a graph problem in the sense of Definition~\ref{def:problem}.
Our lower-bound instances are connected, avoiding the ramifications of the 1v2-cycles conjecture \cite{DBLP:journals/dc/NanongkaiS22,DBLP:conf/stoc/LackiMOS20}.

We first give a constant-round \MPC algorithm that exploits large total memory, matching the space required to store the output in the worst case up to constant factors.
For later use, we slightly generalize the problem to listing only triangles of $G$ belonging to a given set $\mathcal T \subseteq \binom{V}{3}$; choosing $\mathcal T = \binom{V}{3}$ recovers standard triangle listing.

\begin{lemma}
    \label{lem:mpc_triangle_listing}
    Given a graph $G=(V,E)$ and a set $\mathcal T \subseteq \binom{V}{3}$ where $\mathcal T$ can be encoded with at most $S$ words and every machine knows $\mathcal T$ as part of the input. Then listing all triangles in both $G$ and $\mathcal T$ can be accomplished in $\bigO(1/\delta)$ rounds in the $\MPCx{M,S}$ model, with $M \geq |\mathcal T|/n^\delta$ and $S =n^\delta$ for constant $\delta \in (0,1)$.
\end{lemma}

\begin{proof}
    We assign each triangle $\{u,v,w\} \in \mathcal T$ to a machine $M_{uvw}$ in a canonical way, note that we have enough machines such that no machine is responsible for more than $n^\delta$ triangles. We construct the following aggregation task. Each machine $M_{uvw}$ creates three tuples $(\{u,v\},0)$, $(\{v,w\},0)$, $(\{u,w\},0)$. Let $M_{uv}$ be the machine that has an edge $\{u,v\}$ as input, then each $M_{uv}$ makes an aggregation tuple $(\{u,v\},1)$. Now we run the aggregation routine from Lemma \ref{lem:mpc_groupagg} in $\bigO(1/\delta)$ rounds. $M_{uvw}$ reports the triangle $\{u,v,w\}$ if and only if the aggregation for all three edges $\{u,v\},\{v,w\},\{u,w\}$ results to 1.
\end{proof}

For the \NCC lower bound, we consider a class of ``lollipop'' graphs consisting of paths with a dense random graph attached to one end. We show that in the $\NCC(C)$ model there is an information theoretic bottleneck inspired by a lower bound by \cite{DBLP:conf/podc/IzumiG17} for the congested clique, where it is hard for the nodes in the dense part of the graph to list the triangles, and the nodes in the sparse part cannot help because all information on the triangles is ``trapped'' in the dense part. In contrast to \cite{DBLP:conf/podc/IzumiG17}, the lemma gives a lower bound for an all-to-all communication model even on a sparse graph (choosing $n = a^{3}$, we have $\bigO(a^2 + a^3) = \bigO(n)$ edges).

\begin{lemma}
    \label{lem:ncc_lower_bound_lollipop}
    Consider the connected random graphs produced by taking a path on $a^{3}$ vertices, enumerated from $\{1, \dots, a^3\}$ from start to end, and then independently adding each missing edge among pairs of the last $a$ vertices with probability $1/2$. Let $\mathcal{G}$ be the class of (sparse) graphs created with probability $>0$ in this procedure. Any $\NCCx{C}$ algorithm that lists all triangles of a random graph from $\mathcal G$, even when promised that the input is from $\mathcal G$, takes $\Omega(a/C \log n)$ rounds in expectation.
\end{lemma}

\begin{proof}
    Let $G$ be graph from $\mathcal G$ generated as described above. 
    Let $X_e=1$ if one of random edges $e$ exists in $G$ else $X_e = 0$, which is a random variable of Shannon entropy $H(X_e) = 1$ bit. The random variables $X_e$ are independent for distinct random edges $e$ in $G$. Therefore the expected number of bits that must be communicated to a subset $U$ of nodes in order that the nodes in $U$ can list a set of triangles equals the number of distinct edges of these triangles which are yet unknown to any node from $U$.\footnote{This shortened  argument is based on Shannon's concept of information entropy and in particular the source coding theorem \cite{shannon1948_part1,shannon1948_part2}. Loosely, it establishes that a uniquely decodable code that encodes the state of a random variable $X$ requires $H(X)$ bits in expectation. The transcript of a communication protocol that communicates $X$ is a uniquely decodable code, thus the
    transcript of a communication protocol has a length of $H(X)$ bits in expectation.}
    
    Since any three nodes in the dense part of the graph (the last $a$ nodes on the path) form a triangle with probability $1/8$, $G$ has $t \in \Omega(a^3)$ triangles in expectation.     
    In any $\NCC(C)$ protocol $\mathcal A$ each triangle must be reported by some node. Therefore, to solve triangle listing, one of the following two cases must occur. In the first case the first $a^3 - a$ nodes on the path report at least $t/2$ triangles. In the second case the last $a$ nodes on the path, i.e., those in the dense part report $t/2$ triangles. 

    Let us consider case 1. Any set of $t/2$ triangles contains at least \smash{$\tfrac{t^{2/3}}{3 \cdot 2^{1/6}}$} distinct edges by Lemma \ref{lem:edges_of_triangles}, of which at least \smash{$\eta := \tfrac{t^{2/3}}{3 \cdot 2^{1/6}} - a$} are non-path edges. Note that $\eta \in \Omega(a^2)$.    
    Therefore, to output the $t/2$ triangles, the first $a^3-a$ nodes on the path have to learn all $\eta$ edges. Since only the $a$ nodes in the dense part of the graph know these edges, information of $\eta$ bit must have been send by those. Each node can send at most $\bigO(C \log n)$ bits per round, this implies that it takes $\eta/aC \in \Omega(a/C\log n)$ rounds in expectation to communicate this information. 

    Let us consider case 2 where the $a$ nodes in the dense part of the graph output at least $t/2$ triangles. This implies that at least one node $v$ in the dense part reports $t/2a$ triangles, which have at least \smash{$\tfrac{(t/a)^{2/3}}{3 \cdot 2^{1/6}}$} distinct edges by Lemma \ref{lem:edges_of_triangles}. Of these edges $v$ knows at most $a$ itself and at most $a$ are fixed edges on the path, so $v$ must learn at least \smash{$\eta := \tfrac{(t/a)^{2/3}}{3 \cdot 2^{1/6}} - 2a$} edges. Since $\eta \in \Omega(a^{4/3})$ it takes $\Omega(a^{4/3}/C\log n)$ rounds of communication until $v$ knows the edges of the triangles it lists.   
\end{proof}

Combining the preceding lemmas gives the following simulation impossibility result.

\begin{theorem}    \label{thm:sim_mpc_in_ncc_impossible1}
 Let $\mathcal P_a$ be the family of graph problems (Definition \ref{def:problem}) with input graph of arboricity $a \in \Omega(n^{1/3})$. There is a constant $c > 0$, such that $\MPCx{M,S} \nleq_{\mathcal P_a} \NCC\big(\tfrac{MS}{n}\big)$ for $a \geq c \log n \cdot \tfrac{MS}{n}$ (analogously, for total space $MS \leq \tfrac{an}{c \log n}$).
\end{theorem}

\begin{proof}
    Consider a random graph $G$ from the class $\mathcal G$ (see Lemma \ref{lem:ncc_lower_bound_lollipop}) with $n=a^3$ nodes which has arboricity at most $\arb(G) \leq a$. Since the graph is sparse, total space of $MS \in \bigO(a^3)$ is sufficient. In the $\MPCx{a^3/n^{\delta},n^\delta}$ model there is an algorithm $\mathcal A$ to solve the triangle listing problem on any graph $G$ from $\mathcal G$ in $\bigO(1/\delta)$ rounds, for constant $\delta \in (0,1)$, due to Lemma \ref{lem:mpc_triangle_listing}, and exploiting that the set $\mathcal T$ of $\bigO(a^3)$ candidate triangles is known in advance.

    By contrast, in the $\NCCx{C}$ model any algorithm for triangle listing on random graphs $G$ from $\mathcal G$ takes $\Omega(a/C\log n)$ rounds in expectation (and therefore in the worst case), by Lemma \ref{lem:ncc_lower_bound_lollipop}, even when promised that the input is from $\mathcal G$. Any constant-overhead simulation of $\mathcal A$ in the $\NCCx{C}$ model with $C=MS/n$ must necessarily satisfy \smash{$C = MS/n > \tfrac{c' \cdot \delta \cdot a}{\log n}$} for some constant $c' > 0$, in order to match the $\bigO(1/\delta)$ round running time of $\mathcal A$ in the $\MPCx{M,S}$ model. 
    In other words, such a simulation is impossible for $a \geq \tfrac{\log n}{c'\delta} \cdot \tfrac{MS}{n} = c \log n \cdot  \tfrac{MS}{n}$ (where $c = \tfrac{1}{c'\delta}$).
\end{proof}

Note that the impossibility result above relies only on sparse graph instances, exploiting that even a sparse graph can have a dense subgraph with relatively large arboricity ($\Theta(n^{1/3})$ in expectation in the considered instance) that the $\NCC$ model struggles with in terms of the triangle listing problem.

Our next question is whether a constant overhead simulation of $\MPCx{M,S}$ algorithms in the $\NCCx{MS/n}$ model could be feasible, if we consider only graphs of extremely small arboricity? We answer that in the negative as well, showing that a slightly super-logarithmic arboricity already makes it impossible to simulate the fastest possible $\MPC$-algorithm for triangle listing.

\begin{theorem}    \label{thm:sim_mpc_in_ncc_impossible2}
    Let $\mathcal P_a$ be the family of graph problems (Definition \ref{def:problem}) with input graph of arboricity $a \in \omega(\log n) \cap o(n^{1/3})$. Then $\MPCx{M,S} \nleq_{\mathcal P_a} \NCC\big(\tfrac{MS}{n}\big)$.    
\end{theorem}

\begin{proof}
    For this proof we consider an ensemble of $k$ random graphs $G_1, \dots, G_k$ from the class $\mathcal G$, each determined independently with parameter $a$ (see Lemma \ref{lem:ncc_lower_bound_lollipop}). The $G_i$ are pairwise joined together with a single edge between their first and last node on the path respectively (cf.\ construction in the proof of Lemma \ref{lem:ncc_lower_bound_lollipop}) and the node IDs along the whole path are continuous from $1$ to $n$. Let $G$ be the connected graph created in this way. As each graph $G_i$ has $\bigO(a^3)$ nodes and edges, $G$ is sparse and consists of  $k \in \bigO(n/a^3) \in \omega(1)$ many such graphs $G_i$, due to $a \in o(n^{1/3})$. 

    We make two observations on the complexity of triangle listing on $G$ in the different models, derived from the previous proof (of Theorem \ref{thm:sim_mpc_in_ncc_impossible1}). First, Lemma \ref{lem:mpc_triangle_listing} provides a $\bigO(1/\delta)$-round algorithm $\mathcal A$ to solve triangle listing in the $\MPCx{n^{1-\delta}, n^\delta}$ model. Second, the proof of Lemma \ref{lem:ncc_lower_bound_lollipop} shows that communication between nodes of different $G_i$ cannot help overcome the information theoretic barrier in the 
    $\NCC(C)$ model, since the information bottleneck focuses on the dense parts of expected $\Omega(a)$ edges, so the round complexity for solving triangle listing on $G$ is the same as solving it on a single instance of the random graphs $G_i$ from $\mathcal G$, i.e., $\Omega(a/C\log n)$ rounds.
    
    As in the previous proof, we conclude that in order to simulate $\mathcal A$ in $\NCC(C)$, a bandwidth of \smash{$C > \tfrac{c' \cdot \delta \cdot a}{\log n}$} for some constant $c' > 0$ is necessary to match the $\bigO(1/\delta)$ round running time of $\mathcal A$ in $\MPCx{M,S}$. However, we only require $MS= \bigO(n)$ space to solve triangle listing on the sparse graph instance $G$. This implies that we only have $C=MS/n \in O(1)$, bandwidth available in the $\NCC(C)$ model. Therefore, it is impossible to statisfy \smash{$C > \tfrac{c' \cdot \delta \cdot a}{\log n}$} and thus impossible to simulate $\mathcal A$ in $\NCC(C)$ for $a \in \omega(\log n)$.
\end{proof}

\subsection{\MPC Simulation in \NCC}

We now turn to positive simulation results for graph problems, starting with inputs that are sufficiently evenly distributed among the $\NCC$ nodes. To focus on the emulation itself, we first formalize \emph{output compatibility}, which ensures that grouping the outputs of simulated $\MPC$ machines at $\NCC$ nodes gives a valid $\NCC$ output representation.

\begin{definition}[Output compatibility]
\label{def:output_compatibility}
Let $\Pi$ be a graph problem with output representations in $\MPCx{M,S}$ and $\NCCx{MS/n}$, where $M\geq n$. We call $\Pi$ \emph{output-compatible} if, for every instance $(I,S_I)\in\Pi$ with vertex set $V$, every legal $\MPCx{M,S}$ output representation $(o_0,\ldots,o_{M-1})$ of a solution $s\in S_I$ has the following property.
For every partition $(B_v)_{v\in V}$ of the machine IDs $\{0,\ldots,M-1\}$ with $\max_{v\in V}|B_v|\in\bigO(M/n)$, the outputs $(o_i)_{i\in B_v,\;o_i\neq\bot}$ of nodes $v$ is a valid $\NCCx{MS/n}$ output representation of the same solution $s$.
\end{definition}

Output compatibility guarantees that the simulated outputs can be used directly. We later show that efficient redistribution also allows us to satisfy other output conventions, so compatibility is not necessary for an efficient simulation.

When $M\geq n$ and each node initially holds at most $MS/n$ input words, each $\NCCx{MS/n}$ node can emulate its assigned $\bigO(M/n)$ machines in $\bigO(1)$ rounds per round of the $\MPCx{M,S}$ algorithm.

\begin{lemma}
    \label{lem:NCC_dominates_MPC_small_inputs}
    Let $M \geq n$. Then $\MPC(M,S) \leq_{\mathcal P} \NCC(MS/n)$ where $\mathcal P$ are output-compatible problems where the input representation in the $\NCC(MS/n)$ model assigns each node at most $MS/n$ words.
\end{lemma}

\begin{proof}
    In the \NCC model we assume that each node knows the set of IDs and can therefore determine its rank in the order of IDs.
    Each node in the $\NCC(MS/n)$ model simulates $M/n$ machines where the set of simulated machines is determined by the nodes' rank (assuming, w.l.o.g., that $M/n$ is integral). Spreading the input of $MS/n$ words per node among the simulated machines, results in at most $S$ words per simulated machine, which is a valid input representation in the (simulated) $\MPCx{M,S}$ model. Each \NCC node has a communication capacity of $\bigO(S)$ words available per simulated machine and can simulate all assigned machines in $\bigO(1)$ rounds.
\end{proof}

By Lemmas \ref{lem:mpc_sort}, \ref{lem:mpc_groupagg} and \ref{lem:mpc_prefix} there is a $\bigO(1/\delta)$ time algorithm for sorting and group aggregation and computing prefix sums in the $\MPCx{M,S}$ model with $S= n^\delta, M = n$. The lemma above implies that we can simulate these \MPC algorithms efficiently in the \NCC model, when inputs are well distributed in \NCC.

\begin{corollary}
    \label{cor:ncc_sort_agg}
    Sorting, group aggregations and computing prefix sums (see \Cref{sec:mpc_routines}) can be conducted in $\bigO(1/\delta)$ rounds in the $\NCCx{n^\delta}$ model, provided each node has at most $n^\delta$ inputs (sorting keys or aggregation pairs), for constant $\delta \in (0,1)$.
\end{corollary}

We can now show that the $\NCC(C)$ model can simulate the more powerful $\NCC(tC)$ model with an overhead of $\bigO(t^2+t/\delta)$, provided $tC = n^{O(1)}$ and $C \geq n^\delta$ for some constant $\delta \in (0,1)$. A natural approach is to simulate each round by dividing each node's outgoing messages into $t$ batches of at most $C$ messages and sending one batch per round. However, this can violate the restriction of receiving at most $C$ messages per node and round. We show how to construct a schedule that also respects the bandwidth restriction on the receiving side, using the simulated aggregation and prefix-sum routines from the $\MPC$ model.

\begin{lemma}
    \label{lem:ncc_scaling}
    $\NCC(tC) \leq_{t^2+t/\delta} \NCC(C)$ for $t \in \mathbb N$, provided $tC = n^{O(1)}$ and $C \geq n^\delta$ for some constant $\delta \in (0,1)$.
\end{lemma}

\begin{proof}
Let $\mathcal A$ be an algorithm for the $\NCC(tC)$ model and fix one round of its execution. Each node holds at most $tC$ outgoing messages, and at most $tC$ messages are addressed to any node. It suffices to deliver these messages in $\bigO(t^2+t/\delta)$ rounds of $\NCC(C)$, since all local computations of $\mathcal A$ can be performed unchanged.

Let $\rho(v)\in\{0,\ldots,n-1\}$ denote the rank of node $v$ in the globally known set of IDs. We associate with each node $v$ the index range $
J(v):=\{\rho(v)tC,\ldots,(\rho(v)+1)tC-1\}$, i.e., each $v \in V$ is assigned a range of $tC$ consecutive indices where ranges are ordered by node IDs.
Initially, each node $v$ assigns its outgoing messages distinct source indices from $J(v)$. A source index is called occupied if a message is assigned to it. Let $\mathcal M\subseteq\{0,\ldots,ntC-1\}$ denote the set of occupied source indices, and let $r(i)$ denote the intended receiver of the message assigned to $i\in\mathcal M$. Our goal is to compute an injective map
\[
I:\mathcal M\longrightarrow\{0,\ldots,ntC-1\}
\qquad\text{such that}\qquad
I(i)\in J(r(i))
\quad\text{for every }i\in\mathcal M.
\]

Delivering each message $i$ to the node whose index range contains its target index $I(i)$ therefore delivers it to its intended receiver. We first show how each sender computes the target indices of its messages in $\bigO(t/\delta)$ rounds, and then schedule their delivery in $\bigO(t^2)$ rounds while respecting the send and receive bounds.

Partition each index range $J(v)$ into $t$ consecutive blocks of $C$ source indices, numbered $0,\ldots,t-1$. Let $\mathcal M_b$ contain the occupied source indices in the $b$-th blocks across all nodes. For each receiver $r$, we assign its incoming messages consecutive target indices in $J(r)$, ordered first by source block and then by source index.

Fix a source block $b$. We conduct group aggregation by receiver to compute the number $S_{r,b}:=|\{j\in\mathcal M_b:r(j)=r\}|$ of messages addressed to each receiver $r$. We conduct group-based prefix sums ordered by source index to compute the zero-based rank $s_{i,b}:=|\{j\in\mathcal M_b:r(j)=r(i),\ j<i\}|$ of each message $i\in\mathcal M_b$ among the messages with the same receiver. The holder of each message $i$ obtains $S_{r(i),b}$ and $s_{i,b}$. Since each node contributes at most $C$ input tuples, these computations take $\bigO(1/\delta)$ rounds per block.

For each receiver $r$ with $S_{r,b}>0$, the unique message $i\in\mathcal M_b$ with $r(i)=r$ and $s_{i,b}=0$ serves as the representative of the messages in block $b$ addressed to $r$. Its holder reports $S_{r,b}$ to receiver $r$, which records the representative holder's address. For each block, every node sends at most $C$ reports and receives at most one, so this takes $\bigO(1)$ rounds. If $r$ receives no report for block $b$, it sets $S_{r,b}=0$.

After processing all blocks, each receiver $r$ locally computes the offset $Q_{r,b}:=\sum_{a=0}^{b-1}S_{r,a}$, which counts the messages addressed to $r$ in earlier source blocks. For each block $b$, receiver $r$ sends this offset to its recorded representative holder, if one exists. Every node sends at most one such reply and receives at most $C$. 

For each message $i\in\mathcal M_b$, its holder contributes the aggregation tuple $(r(i),Q_{r(i),b})$ if $s_{i,b}=0$, and $(r(i),0)$ otherwise. We group these tuples by their first component and sum their second components. For each receiver $r$, the representative contributes $Q_{r,b}$ and all other tuples for $r$ contribute zero, so their sum is $Q_{r,b}$. The aggregation routine returns this sum to every node that contributed a tuple for $r$. Thus, the holder of every message $i\in\mathcal M_b$ learns the offset $Q_{r(i),b}$ needed to compute its target index.

The holder of each message $i\in\mathcal M_b$ can now compute its target index
\[
I(i):=\rho(r(i))tC+Q_{r(i),b}+s_{i,b}.
\]
For a fixed receiver $r$, the offsets $Q_{r,b}+s_{i,b}$ enumerate its incoming messages consecutively from zero. Since at most $tC$ messages are addressed to $r$, their target indices are distinct and belong to $J(r)$. The ranges $J(r)$ are disjoint, so $I$ is injective.

Each block requires a constant number of aggregation and prefix-sum computations on at most $C$ input tuples per node, together with $\bigO(1)$ rounds of direct communication. By the small-input simulation of Lemma~\ref{lem:NCC_dominates_MPC_small_inputs} and the routines of Lemmas~\ref{lem:mpc_groupagg} and~\ref{lem:mpc_group_prefix}, each computation takes $\bigO(\log(nC)/\log C)=\bigO(1/\delta)$ rounds. The total time to compute $I$ is therefore $\bigO(t/\delta)$.

Finally, regard the $t$ blocks of each range $J(r)$ as target blocks. We allocate one round to each pair $(b,b')\in\{0,\ldots,t-1\}^2$ of a source block and a target block. In the round for $(b,b')$, each node sends each of its original outgoing messages $i\in\mathcal M_b$ to $r(i)$ if $b'C\leq I(i)-\rho(r(i))tC<(b'+1)C$, meaning that its target index lies in the $b'$-th target block of its receiver.

Each node sends at most $C$ messages, since only its $b$-th source block contributes. Each receiver receives at most $C$ messages, since their target indices are distinct by injectivity of $I$ and belong to its $b'$-th target block, which contains $C$ indices. Every message belongs to exactly one source block and has its target index in exactly one target block, so every message is delivered exactly once. The delivery therefore takes $t^2$ rounds. Together with the $\bigO(t/\delta)$ rounds required to compute $I$, this simulates one round of $\mathcal A$ in $\bigO(t^2+t/\delta)$ rounds of $\NCC(C)$.
\end{proof}

The lemma above combined with Corollary~\ref{cor:ncc_sort_agg} allows us to perform the basic \MPC routines with larger data sets of at most $t n^\delta$ inputs per node. These routines take $\bigO(1/\delta)$ rounds in $\NCCx{t\cdot n^\delta}$, and each round can be simulated in $\NCCx{n^\delta}$ with an overhead of $\bigO(t^2+t/\delta)$.

\begin{corollary}
    \label{cor:ncc_sort_agg_large}
    Sorting, group aggregations and computing prefix sums (see \Cref{sec:mpc_routines}) with at most $t\cdot n^\delta$ inputs per node can be conducted in $\bigO(t^2/\delta+t/\delta^2)$ rounds in the $\NCCx{n^\delta}$ model, for constant $\delta\in(0,1)$ and any integer $1\leq t\leq n^{1-\delta}$.
\end{corollary}

\subparagraph*{Deterministic Identification in \NCC}

We will now consider the \emph{identification problem} in the \NCC model, where every node is either playing or learning or neither \cite{DBLP:conf/spaa/AugustineGGHSKL19}. The learning and playing nodes know their own role and the IDs of their neighbors. Each playing node $u$ knows a set $L(u)$ of its neighbors that contains all its learning neighbors. The learning nodes do not know anything about their neighbors' roles. To solve the problem, all learning nodes must identify all of their playing neighbors.

What makes this problem challenging is that a learning node could have a large number of playing neighbors, so that directly sending this information from playing to learning nodes is not efficient. We use simulated \MPC aggregation routines to identify the non-playing neighbors of each learning node by successively refining intervals of IDs. Our approach offers a deterministic alternative to the randomized approach of \cite{DBLP:conf/spaa/AugustineGGHSKL19}.

\begin{theorem}
    \label{thm:identification}
    Assume each learning node has at most $n^{(1-\varepsilon)\delta}$ neighbors that are \emph{not} playing and each playing node $u$ has $|L(u)|\leq n^{\delta}$, for constants $\varepsilon,\delta \in (0,1)$. There is a deterministic algorithm to solve the identification problem in $\bigO(1/(\varepsilon\delta^2))$ rounds in the $\NCCx{n^{\delta}}$ model.
\end{theorem}

\begin{proof}
    Let $B=\lfloor n^{\varepsilon\delta}\rfloor$, which is at least two for sufficiently large $n$, and let $d$ be the smallest power of $B$ strictly larger than the largest ID. Since the largest ID has a $\poly(n)$ upper bound, $d$ is $\poly(n)$ as well. At granularity $i$, we divide $[0,d)$ into $B^i$ equally sized sub-intervals $[d_{ij},d_{i,j+1})$, numbered by $j$ in ascending order. For each interval $[d_{i-1,k},d_{i-1,k+1})$ retained at the previous level, a learning node $v$ must determine which of its $B$ sub-intervals contain a non-playing neighbor of $v$. Initially, each learning node retains $[0,d)$.
    
    At level $i$, each playing node $u$ chooses the unique $j$ with $ID(u)\in[d_{ij},d_{i,j+1})$. For every $v\in L(u)$, it contributes $(g_{jv},-1)$, where $g_{jv}:=j\circ ID(v)$ denotes the aggregation group (see \Cref{sec:mpc_routines}). Only learning nodes use the sums for groups addressed to them.
    For each sub-interval it queries, a learning node $v$ contributes $(g_{jv},t_{jv})$, where $t_{jv}$ is its total number of neighbors in $[d_{ij},d_{i,j+1})$. Every playing neighbor $u$ in this interval contributes $-1$ to the group, since $v\in L(u)$. Thus the group sum $s_{jv}$ returned to $v$ counts its non-playing neighbors in the interval. Node $v$ retains precisely the queried intervals with $s_{jv}>0$.
    
    All queries at level $i$ can be answered by a single group aggregation using Corollary~\ref{cor:ncc_sort_agg}, provided each node contributes at most $n^\delta$ tuples. Each playing node $u$ contributes exactly $|L(u)|\leq n^\delta$ tuples. At level $i=1$, each learning node $v$ queries all $B\leq n^\delta$ intervals. It retains at most $n^{(1-\eps)\delta}$ intervals, since each retained interval contains a distinct non-playing neighbor of $v$.
    
    If $v$ has retained at most $n^{(1-\eps)\delta}$ intervals at level $i$, it queries the $B$ children of each retained interval at level $i+1$. It therefore contributes at most $B n^{(1-\eps)\delta}\leq n^\delta$ tuples. Since the retained intervals at level $i$ contain all non-playing neighbors of $v$, these queries identify all intervals containing such a neighbor at level $i+1$. Again, there are at most $n^{(1-\eps)\delta}$ of them.

    When the sub-intervals reach size 1, each learning node knows its non-playing, and thus its playing neighbors. Since $d=\poly(n)$ and $\log B=\Theta(\eps\delta\log n)$, the intervals reach size 1 after $\log_B d=\bigO(1/(\eps\delta))$ iterations. Each iteration consists of a group aggregation taking $\bigO(1/\delta)$ rounds by Corollary~\ref{cor:ncc_sort_agg}, resulting in the claimed overall round complexity.
\end{proof}

We can parameterize the problem by the number of non-playing neighbors of each learning node and the size of each playing node's set $L(u)$. The theorem above gives us a solution by dilating the number of rounds to handle a larger throughput.

\begin{corollary}
    \label{cor:identification}
    Let $p$ be the maximum number of non-playing neighbors of a learning node, and let $\ell$ be the maximum of $|L(u)|$ over playing nodes $u$. Take either maximum to be zero if its node set is empty.
    There is a deterministic algorithm to solve the identification problem in $\bigO\big((\tfrac{p+\ell}{n^{(1-\eps)\delta}}+1)^2 \tfrac{1}{\eps^2\delta^3}\big)$ rounds in the $\NCCx{n^{\delta}}$ model, for any constants $\eps,\delta \in (0,1)$.
\end{corollary}

\begin{proof}
    If $p+\ell \leq n^{(1-\eps)\delta}$, the claim follows directly from Theorem~\ref{thm:identification}.

    Otherwise, choose $\delta'>\delta$ such that $n^{\delta'-\eps\delta}=p+\ell$, and set $\eps'=\eps\delta/\delta'$. Then $\eps'\in(0,1)$ and $p+\ell=n^{(1-\eps')\delta'}$.

    If $\delta'<1$, we use the construction of Theorem~\ref{thm:identification} with parameters $\eps'$ and $\delta'$. Its refinement factor is $\lfloor n^{\eps'\delta'}\rfloor=\lfloor n^{\eps\delta}\rfloor$, so it uses $\bigO(1/(\eps\delta))$ iterations. Each iteration takes $\bigO(1/\delta')\subseteq\bigO(1/\delta)$ rounds in $\NCCx{n^{\delta'}}$, giving $\bigO(1/(\eps\delta^2))$ rounds in total. This analysis also applies when the auxiliary parameters depend on $n$, since $\eps'\delta'=\eps\delta$ and $\delta'>\delta$.

    If $\delta'\geq1$, each playing node $u$ sends its ID to every $v\in L(u)$ in one round of $\NCCx{n^{\delta'}}$. Each learning node receives the IDs of precisely its playing neighbors, and every node sends and receives at most $n-1$ messages.

    In either case, the problem can therefore be solved in $\bigO(1/(\eps\delta^2))$ rounds in $\NCCx{n^{\delta'}}$. Set $t=\lceil(p+\ell)/n^{(1-\eps)\delta}\rceil$, so that $n^{\delta'}\leq t n^\delta$. By Lemma~\ref{lem:ncc_scaling}, each such round can be simulated in $\bigO(t^2+t/\delta)$ rounds in $\NCCx{n^\delta}$. The total round complexity is thus $\bigO((t^2+t/\delta)/(\eps\delta^2))\subseteq\bigO(t^2/(\eps\delta^3))$, which implies the claimed bound because $t\leq(p+\ell)/n^{(1-\eps)\delta}+1$ and $\eps<1$.
\end{proof}

\subparagraph*{\MPC Simulation in \NCC on Bounded Arboricity Graphs}

Our prime application for the identification problem is an efficient edge-orientation of a graph $G$ such that the out-degree is not too much larger than the arboricity of $G$ (see Appendix \ref{sec:graph_properties} for a characterization of the arboricity of a graph). In the \NCC model, a representation of a valid solution for this problem stores the list of out-neighbors of such an orientation.

To compute this orientation we leverage fast edge orientation \MPC algorithms that we introduce in Appendix \ref{sec:mpc_routines_forest_decomp}.
The process for edge-orientation is described in Lemma \ref{lem:mpc_orientation} and proceeds in phases. In each phase, we orient every edge incident to vertices whose current degree is $\bigO\big(\beta\cdot\arb(G)\big)$, then remove those edges. We repeat this process on the residual graph until all edges are oriented. 

Unlike the \MPC\ routine of Lemma \ref{lem:mpc_orientation}, aggregation alone cannot reveal which edges must be oriented: a vertex that appears low-degree in the residual graph may still have high overall degree, and it cannot locally tell which of its incident edges are already oriented. Consequently, we must first solve a node-identification subproblem before each orientation phase. We show the following lemma.

\begin{lemma}
    \label{lem:ncc_orientation}
    Let $G$ be a graph with at least one edge. An orientation of $G$ with out-degree less than $2\beta\arb(G)$ can be computed in $\bigO\big((\tfrac{\beta\cdot\arb(G)}{n^{(1-\eps)\delta}}+1)^2\cdot\tfrac{1}{\eps^2\delta^3}\cdot\log_\beta n\big)$ rounds in the $\NCCx{n^\delta}$ model, for any constants $\eps,\delta\in(0,1)$ and $\beta\geq2$.
\end{lemma}

\begin{proof}
    If $2\beta\arb(G)\geq n$, orient every edge toward the endpoint with the smaller ID. This requires no communication and gives out-degree at most $n-1<2\beta\arb(G)$. We may therefore assume that $2\beta\arb(G)<n$.

    We carry out the edge-orientation procedure while maintaining two invariants after every phase: (i) each active node $v$ knows its current degree $d(v)$ in the residual graph; (ii) each node that has already been removed knows its fewer than $2\beta\arb(G)$ outgoing neighbors. Initially, every node is active and knows its degree. In each phase, we orient all edges incident to active nodes with fewer than $2\beta\arb(G)$ residual neighbors and remove these nodes. The difficulty is that a node can determine that its residual degree is below the threshold but cannot identify which of its incident edges remain unoriented. We therefore frame this task as an instance of the identification problem.

    Playing nodes are those removed in previous phases, while learning nodes are the active nodes whose residual degree is below $2\beta\arb(G)$. Every edge between a playing node and an active node was oriented toward the active node when the playing node was removed. Thus, by invariant (ii), each playing node knows a superset of its learning neighbors consisting of fewer than $2\beta\arb(G)$ outgoing neighbors. Conversely, each learning node has fewer than $2\beta\arb(G)$ neighbors that are still active, i.e., not playing. This gives an identification instance with $p,\ell\in\bigO(\beta\arb(G))$. Corollary~\ref{cor:identification} solves it in $\bigO\big((\tfrac{\beta\arb(G)}{n^{(1-\eps)\delta}}+1)^2\tfrac{1}{\eps^2\delta^3}\big)$ rounds. Afterwards, every learning node knows its playing neighbors and can therefore identify its incident edges that still need to be oriented.

    To orient these edges, we must resolve ties when both endpoints are learning nodes removed in the same phase. Either direction is acceptable, but both endpoints must agree. We frame this as an aggregation task. For every unoriented edge $\{u,v\}$ incident to a learning node $u$, it creates $(u,v)$ if $ID(u)\leq ID(v)$ and $(v,u)$ otherwise. Group aggregation using these pairs as group identifiers gives each contributing learning node the count $C_{u,v}\in\{1,2\}$. Since each node contributes fewer than $2\beta\arb(G)$ tuples, the same simulation as in Corollary~\ref{cor:ncc_sort_agg_large} performs this aggregation in $\bigO\big((\tfrac{\beta\arb(G)}{n^\delta}+1)^2\tfrac{1}{\delta^2}\big)$ rounds. If $C_{u,v}=1$, the edge is directed away from its learning endpoint; if $C_{u,v}=2$, it is directed toward the endpoint with the smaller ID. Every learning node now knows its outgoing neighbors, establishing invariant (ii) for the nodes removed in this phase.

    To restore invariant (i), each node that remains active must learn its new residual degree. Every learning node $u$ creates a pair $(v,-1)$ for each of its neighbors $v$ in the residual graph at the start of the phase. Each node $v$ that remains active contributes $(v,d(v))$. Aggregating by the first component gives each such node its previous degree minus the number of its neighbors removed in this phase, which is its new residual degree. Since each node again contributes at most $2\beta\arb(G)$ tuples, the same aggregation bound applies.

    Since $n^\delta\geq n^{(1-\eps)\delta}$ and $\eps,\delta\in(0,1)$, both aggregation costs are bounded by the identification cost. Thus each phase takes $\bigO\big((\tfrac{\beta\arb(G)}{n^{(1-\eps)\delta}}+1)^2\tfrac{1}{\eps^2\delta^3}\big)$ rounds. Every residual graph has average degree below $2\arb(G)$, so each phase leaves at most a $1/\beta$-fraction of its previously active nodes. Hence $\bigO(\log_\beta n)$ phases suffice, as in Lemma~\ref{lem:mpc_orientation}, proving the claimed round bound.
\end{proof}

Finally, we can show how to use the tools above to simulate $\MPC$ in $\NCC$. We can leverage the edge orientation developed for the $\NCC$ model to make one endpoint of each edge responsible for managing that edge such that no node is responsible for more than $2\beta\arb(G)$ edges. With such a bound on the input size, we show that each node in the $\NCCx{MS/n}$ model can simulate $M/n$ machines of a given $\MPCx{M,S}$ algorithm.

\begin{theorem}
    \label{thm:ncc_sim_mpc}
    Let $\mathcal P_a$ be the family of output-compatible graph problems with input graph of bounded arboricity $a \leq \tfrac{MS}{n^{1+\delta}}$ (or analogously, minimum total memory of $MS \geq a \cdot n^{1+\delta}$),
    for any constant $\delta \in (0,1)$ and $M \geq n$. Then $\MPCx{M,S} \,\leq_{\mathcal P_a}\! \NCC\big(\tfrac{MS}{n}\big)$, with an additive (constant in $n$) overhead of $\bigO(1/\delta^4)$.
\end{theorem}

\begin{proof}
    Let $\mathcal A$ be an $\MPCx{M,S}$ algorithm for a problem in $\mathcal P_a$, with round complexity $T_{\mathcal A}$, and write $C:=MS/n$. We first compute an orientation with at most $C$ outgoing edges per node.

    For an edgeless graph, this preprocessing is unnecessary. Otherwise, $a\geq1$, so the arboricity assumption implies $C\geq n^\delta$. If $C\geq n$, we can orient every edge locally toward the endpoint with the smaller ID, giving out-degree at most $n-1\leq C$. It therefore remains to consider $n^\delta\leq C<n$.

    Let $\delta'$ satisfy $C=n^{\delta'}$, so that $\delta\leq\delta'<1$. We choose $\beta:=\lfloor n^{\delta/2}/4\rfloor$ and $\eps:=\delta/(2\delta')$. For sufficiently large $n$, we have $\beta\geq2$ and $\log_\beta n\in\bigO(1/\delta)$. Moreover, $\eps\in(0,1)$ and $(1-\eps)\delta'=\delta'-\delta/2$. Since $a\leq n^{\delta'-\delta}$, these choices give $\beta a/n^{(1-\eps)\delta'}\leq1/4$ and $2\beta a\leq C/(2n^{\delta/2})\leq C$.

    Using the construction of Lemma~\ref{lem:ncc_orientation} with capacity exponent $\delta'$ and parameter $\eps$, we obtain an orientation with out-degree less than $2\beta a\leq C$ in
    \begin{align*}
        \bigO\big((\tfrac{\beta a}{n^{(1-\eps)\delta'}}+1)^2\cdot\tfrac{1}{\eps^2\delta'^3}\cdot\log_\beta n\big)
        &\subseteq \bigO\big(\tfrac{1}{\eps^2\delta'^3}\cdot\tfrac{1}{\delta}\big)\\
        &= \bigO\big(\tfrac{1}{\delta^3\delta'}\big)
        \subseteq \bigO(1/\delta^4)
    \end{align*}
    rounds. The construction and its analysis also apply if $\delta'$ depends on $n$, since $\delta'\geq\delta$ and the exponent $\eps\delta'=\delta/2$ is fixed.

    We now make each node responsible only for its outgoing edges, together with their input labels. This assigns each edge to exactly one endpoint and gives each node at most $C$ edge records, yielding a valid input representation for the simulated $\MPCx{M,S}$ algorithm (Definition~\ref{def:representation}). Since $M\geq n$, Lemma~\ref{lem:NCC_dominates_MPC_small_inputs} allows us to simulate $\mathcal A$ in $\NCC(C)$ with constant overhead per simulated round, assigning $\bigO(M/n)$ machines to each node.

    By output compatibility (Definition~\ref{def:output_compatibility}), the outputs of the simulated machines, grouped at their simulating nodes, form a valid $\NCC(C)$ output representation of the same solution. The total round complexity is therefore $\bigO(T_{\mathcal A}+1/\delta^4)$, proving the claimed additive overhead.
\end{proof}

Theorem~\ref{thm:ncc_sim_mpc} complements the preceding impossibility results.
With $C=MS/n$, it gives constant-overhead simulations for $M\geq n$ and $a\leq C/n^\delta$, whereas Theorem~\ref{thm:sim_mpc_in_ncc_impossible1} rules them out for $a\geq cC\log n$ in the range $a\in\Omega(n^{1/3})$.
The arboricity thresholds thus differ by a factor $\bigO(n^\delta \log n)$, where $\delta>0$ can be an arbitrarily small constant.
Theorem~\ref{thm:sim_mpc_in_ncc_impossible2} further shows that the obstruction persists for $a\in\omega(\log n)\cap o(n^{1/3})$, already on sparse graphs with $MS=\bigO(n)$.

An $n$-node graph of arboricity $a$ can have $\Theta(an)$ edges, so storing such inputs requires $\Omega(an)$ total memory in the worst case.
Our sufficient condition $MS\geq an^{1+\delta}$ exceeds this worst-case input-space requirement by only a factor $n^\delta$, for an arbitrarily small constant $\delta>0$.

In particular, Theorem \ref{thm:ncc_sim_mpc} is interesting in light of the 1v2-cycles problem, which is conjectured to be impossible in sub-logarithmic time in the strongly sublinear regime $S=n^\delta$, even for polynomially large total memory $MS = \poly(n)$ \cite{DBLP:conf/stoc/LackiMOS20, DBLP:journals/dc/NanongkaiS22}. In the (unlikely) event that a $o(\log n)$ solution for \MPC with total memory $n^{1+\delta} \leq SM \leq n^{2-\varepsilon}$ for $\delta, \varepsilon > 0$ is found, this would imply the same round complexity for the sublinear $\NCCx{MS/n}$ model. This would be surprising as pointer doubling with round complexity $\Theta(\log n)$ is considered the best possible in sublinear \NCC (arguably another reason to believe the conjectured $\Omega(\log n)$ lower bound in sublinear \MPC).

\begin{lemma}
\label{lem:ncc_output_redistribution}
Let $C=MS/n=n^\delta$ for a constant $\delta\in(0,1)$ and $M\geq n$. Suppose the output of an $\MPCx{M,S}$ algorithm is stored at $\NCCx{C}$ nodes, each holding the outputs of $\bigO(M/n)$ simulated machines. Let each part of the output be assigned to $\bigO(1)$ destination nodes whose IDs are known to its current holder, with at most $\bigO(C)$ words assigned to each destination, counting all copies. Then the prescribed output distribution can be obtained deterministically in $\bigO(1/\delta)$ rounds.
\end{lemma}

\begin{proof}
Each simulated machine holds at most $S$ words of output, so each $\NCC$ node initially holds $\bigO((M/n)S)=\bigO(C)$ output words. Since each output part has only constantly many destinations, each node needs to send $\bigO(C)$ words, including the constant overhead for addressing and identifying the transmitted parts. By assumption, each node also receives at most $\bigO(C)$ words.

Consequently, the entire redistribution can be performed in one round of $\NCCx{tC}$ for a sufficiently large constant $t$. By Lemma~\ref{lem:ncc_scaling}, this round can be simulated deterministically in $\NCCx{C}$ in $\bigO(t^2+t/\delta)=\bigO(1/\delta)$ rounds, since $t$ is constant and $\delta\in(0,1)$.
\end{proof}

For the vertex coloring, maximal independent set, and maximal matching applications in Section~\ref{sec:overview}, each vertex needs only its color, membership bit, or matched neighbor's ID (or an unmatched marker), respectively. These $\bigO(1)$-word outputs can be delivered in $\bigO(1/\delta)$ additional rounds by Lemma~\ref{lem:ncc_output_redistribution}. For matching, this information also determines the selection status of all incident edges, regardless of the vertex's degree. Triangle listing already permits each triangle to be reported by any node, so the inherited output representation is directly valid.

\section{Simulation Impossibility for Decision Problems}
\label{sec:sim_mpc_in_ncc_dec_problems_impossible}

The triangle-listing lower bounds in Subsection~\ref{sec:sim_mpc_in_ncc_impossible} exploit a large output.
We now ask whether graph problems with small outputs admit simulations with substantially less total bandwidth.
We show that even for decision problems the answer is (conditionally) no using a circuit-complexity-based argument in the polylogarithmic \MPC regime.
More precisely, we prove that a more efficient simulation of $\MPC$ in $\NCC$ than what we provide in Theorem \ref{thm:ncc_sim_mpc}, would also provide some answers towards the longstanding and notoriously intractable question \smash{$\NClin \stackrel{?}{=}\NC$} for decision problems \cite{DBLP:conf/mfcs/Valiant77, DBLP:conf/innovations/GolovnevKW21}. Loosely speaking, if these two classes were equal, any polynomial size circuit of polylogarithmic depth could be compressed to linear size while keeping polylogarithmic depth which is currently not known. 

\subparagraph*{Preliminaries} We first give some more detailed definitions  and nomenclature. We interpret a decision problem $\Pi$ as a language $\Pi \subseteq \{0,1\}^*$, where $w \in \Pi$ iff $w$ encodes a yes-instance of the problem. For instance, assume that $\Pi$ is the triangle decision problem, i.e., given that $w$ is an encoding of edges of some graph $G$, then $w \in \Pi$ iff $G$ has a triangle. In our distributed models, at least one machine decides yes on instance $w$, if and only if $w \in \Pi$.

A circuit $C_n$ is a directed acyclic graph (DAG) of wires (edges) and gates (nodes) where inner nodes correspond to boolean gates, AND, OR, NOT, with fan-in (i.e., number of input wires of a gate) at most 2. The source nodes of the DAG correspond to $n$ input gates with fan-in 0, each representing a bit of the input $w \in \{0,1\}^n$. The fan-out is the maximum number of gates the output wire of a given gate feeds into. A circuit $C_n$ for a decision problem has a dedicated output wire (or equivalently an output gate) and a family $\{C_n\}$ decides $\Pi$, if for any $n \in \mathbb N$ and $w \in \{0,1\}^n$, $C_n$ outputs 1 iff $w \in \Pi$.
The family $\{C_n\}$ is called \Pclass-uniform, if there is a Turing machine (TM) that outputs (``generates'') $C_n$ in polynomial steps on input $1^n$.\footnote{The reason for the unary encoding of $n$ is purely formal, to keep the TM polynomial in the length of the input.} A stricter variety is $\Lclass$-uniformity where the TM is $\log n$-space restricted.\footnote{Log-space restricted means that the Turing machine has a read-only input tape, a write-only output tape and a $\bigO(\log n)$ read-write tape as ``work memory''. $\Lclass$-uniformity implies $\Pclass$-uniformity because a terminating log-space Turing machine has only $n^{\bigO(1)}$ configurations and cannot revisit one.} 

Then $\NC$ is the set of languages $\Pi \subseteq \{0,1\}^*$ that can be decided by a uniform family of circuits $\{C_n\}$ of size $\bigO(n^\ell)$ and depth $\bigO(\log^k n)$ for constants $k,\ell \in \mathbb N$.
The fan-out of $\{C_n\}$ is assumed unbounded, however a circuit of depth $\bigO(\log n)$ can always be transformed to fan-out 1 with the same depth and $\poly(n)$ size blowup. It is also possible to bound the fan-out of a circuit with depth $\bigO(\log^k n)$ to 2 at the cost of a constant factor size increase, and depth increase by a factor of $\bigO(\log n)$.\footnote{We can turn a size $s$ circuit with fan-in 2 and depth $d \in \bigO(\log n)$ into a propositional formula $F$ of size $s' = s \cdot 2^d = \poly(n)$ and depth $d$ \cite{DBLP:books/daglib/0028687}. $F$ can be balanced to depth $\log(s') \in \bigO(\log n)$ using Spira's theorem \cite{Spira71}. For depth $\bigO(\log^k n)$ we can bound the fan-out to 2 by replacing high fan-out gates with balanced ``copy trees'', which increases depth by a factor at most $\bigO(\log n)$.} 
The class $\NCx k$ is the subset of $\NC$ with depth $\bigO(\log^k n)$. The class $\NClinx k$ is the subset of problems in $\NCx k$ that admit a circuit size of just $\bigO(n)$. Note that the intractability of the question \smash{$\NClin \stackrel{?}{=}\NC$} is not affected by the type of uniformity (\Pclass- or \Lclass-) of circuits as long as we consistently settle for either one (\Lclass-uniformity in our case).

\subsection{Distributed Algorithms to Uniform Circuit Transformation} 
To relate distributed algorithms to circuit complexity, we transform a given algorithm into a circuit with bounded size and depth (equivalently: we give a circuit that simulates a distributed algorithm). A main obstacle is that in the distributed model, communication may depend on the actual input instance, whereas in a uniform circuit the wiring pattern may depend only on the length of the input. We resolve this by replacing data-dependent communication with an oblivious schedule whose pattern depends only on the instance size and is computable by a log-space Turing machine (i.e., $\Lclass$-uniform).

\begin{definition}
    We call a \MPCx{M,S} (or $\NCCx{C}$) algorithm input-oblivious, if machine (or node) $i$ in round $r$ sends only to a subset $\mathcal N_{ir}$ of others, 
    and all the sets $\mathcal N_{ir}$ can be written by a Turing machine on a (write-only) output tape using $\bigO(\log n)$ space of (read-write) working memory.
\end{definition}

Regarding notation in this and the following claims and proofs: $C$ denotes the \textit{capacity} in the $\NCC(C)$ model whereas adding sub- or superscripts (e.g., $C_n$, $C^{ij}$) will always denote \textit{circuits}.

\begin{lemma}
    \label{lem:mpc_oblivious}
    Any \MPCx{M,S}-algorithm $\mathcal A$ that takes $t$ rounds can be simulated by an input-oblivious algorithm that computes the same result and takes $t \cdot \polylog(M S)$ rounds.
\end{lemma}

\begin{proof}
    Our approach makes $\MPCx{M,S}$ algorithms input-oblivious by routing all communication through a fixed sorting network. A sorting network is a comparator circuit with a wiring pattern independent of the input. Each comparator takes two keys and outputs them in non-decreasing order on its two output connections. We fix an $\Lclass$-uniform family ${C_m}$, so that given $m$ one can generate the description of $C_m$ (equivalently, its comparators and connections) using $O(\log m)$ workspace. For $m$ keys there are such networks of depth $\polylog (m)$ with $m$ connections per layer (e.g., the bitonic network \cite{DBLP:conf/afips/Batcher68} or the AKS network \cite{DBLP:conf/stoc/AjtaiKS83}). The $\Lclass$-uniformity lets each machine locally compute the portion of $C_m$ it needs within its memory bounds. We then replace the input-dependent communication of $\mathcal{A}$ by this predetermined pattern. First, we implement input-oblivious sorting in $\MPCx{M,S}$ using $C_m$, and then use that sorting routine to deliver arbitrary point-to-point messages.

    Let $m:=MS$ be the size of the total memory (in machine words) and identify these with the $m$ connections of $C_m$ at layer $0$ they carry the inputs, and after the last layer they carry the outputs in sorted order. Each machine holds $S$ keys and can evaluate any comparator whose two inputs reside on that machine. We implement each layer of $C_m$ in \MPC as follows. (i) We fix a canonical assignment of that layer’s $MS/2$ comparators to machines with at most $S/2$ comparators per machine. Due to the $\Lclass$-uniformity of $\{C_m\}$ each machine can locally determine its assigned comparators. (ii) Locally apply the resulting compare-and-swap operations of the assigned comparators on the local memory. (iii) Route the resulting keys to the machines in the next layer, as dictated by the fixed pattern of $C_m$ (which can be determined locally due to the $\Lclass$-uniformity of $\{C_m\}$). Repeating this for the $\polylog(m)$ layers of $C_m$ sorts all $m$ keys. Each machine only exchanges messages along some connection in $C_m$ so this communication pattern can be computed by the Turing machine that computes $C_m$.

    Now we use sorting for the input-oblivious message delivery in $\mathcal A$. First consider the ``clean'' case in which every machine sends exactly $S$ messages. Because $\mathcal A$ is an $\MPCx{M,S}$ algorithm, each destination receives at most $S$ messages. If we sort all messages globally by destination ID, the block for destination $i$ is contiguous and has length exactly $S$. So every message for machine $i$ lies within positions $[iS,(i+1)S-1]$ of the global memory, hence machine $i$ finds all its messages in its local memory with no further routing.
    This argument breaks if some senders emit fewer than $S$ messages, as the messages for machine $i$ no longer align with its memory. We therefore reduce this to the balanced case by padding to exactly $S$ messages per machine.

    Each machine prepares up to $2S$ messages by taking its actual messages (at most $S$), keyed by destination ID, and adding exactly $S$ dummy messages addressed to itself. We pack two logical messages per machine word (using the $\Theta(\log n)$-bit word size), so per-machine space remains $S$ words. After a global sort by destination ID, the messages for any destination $j$ form a contiguous block of size at most $2S$, and hence lie on at most two consecutive machines (as each holds up to $2S$ messages). Due to the padding messages, each machine’s local memory contains messages for at most three destination IDs.
    Then, each machine coordinates only with its predecessor and successor machine so that every destination retains exactly $S$ messages (keeping all real messages and discarding any excess dummies). The first $M/2$ machines now hold $MS$ messages in destination-sorted order. Finally, deliver them via a fixed pattern in which machine $i$ forwards its local messages accordingly to machines $2i$ and $2i+1$. After this step, each destination machine holds its $S$ messages in its local memory. 
    
    The round complexity of $\mathcal A$ increases only by a factor of $\polylog(m)$. As we can describe the communication between machines in the transformed algorithm using the log-space description of $C_m$ and simple arithmetic (message exchanges of machine $i$ with machines $i-1, i+1,2i, 2i+1$) it is computable by a log-space Turing machine.
\end{proof}

Using our simulation result for $\MPCx{M,S}$ algorithms with a balanced problem input from Lemma \ref{lem:NCC_dominates_MPC_small_inputs}, we can simulate the input-oblivious message delivery procedure above in the $\NCCx{C}$ model (with $MS = nC$) and obtain the following result as corollary.

\begin{corollary}
    \label{cor:ncc_oblivious}
    Any $\NCC(C)$-algorithm $\mathcal A$ that takes $t$ rounds can be simulated by an input-oblivious algorithm that computes the same result and takes $t \cdot \polylog(nC)$ rounds.
\end{corollary}

We can now give circuit simulations starting with $\MPC$ algorithms. We say that a $\MPCx{M,S}$-algorithm $\mathcal A$ decides a language $\Pi \subseteq \{0,1\}^*$ in $t$ rounds, if for any $n \in \mathbb N$ there are parameters $M,S$ with $MS \geq n$ such that for any input $w \in \{0,1\}^n$ distributed on the machines, at least one machine decides yes within $t$ rounds if and only if $w \in \Pi$. We say that a circuit family $\{C_n\}$ simulates an algorithm $\mathcal A$ that decides $\Pi$, if $C_n$ outputs 1 on input $w \in \{0,1\}^n$ if and only if $w \in \Pi$.
Furthermore, we say that an algorithm in the \MPC model is $\P$-bounded if all local computation on the machines are polynomial in $S$. In a sense we are forced to assume $\Pclass$-boundedness, because if nodes could, for instance, solve \NP-complete problems locally, then this would in general preclude a polynomial-size circuit simulation of this algorithm (unless $\Pclass=\NP$).

\begin{lemma}
    \label{lem:mpc_to_circuit}
    Let $\Pi \subseteq \{0,1\}^*$ be a decision problem and let $\mathcal A$ be a \P-bounded $\MPCx{M,S}$-algorithm that decides $\Pi$ in $t$ rounds. For $M,S,t \in \bigO(\poly(n))$ there is a \Lclass-uniform family $\{C_n\}$ of circuits with fan-out at most 2, size $\bigO\big(tM\cdot\poly(S\log n)\big)$ and depth $\bigO\big(t\cdot\poly(S\log n)\big)$ that solves $\Pi$ (i.e., $\{C_n\}$ simulates $\mathcal A$).
\end{lemma}

\begin{proof}
    In \MPC, a single machine is modeled as $\bigO(S \log n)$-space bounded Turing machine $T^{ij}$.  Therefore, the computation that machine $i$ conducts in round $j$ on its current memory content, can be represented by a boolean circuit $C^{ij}$ with $\bigO(S \log n)$ input and output gates. The input gates represent the machine memory (bit by bit) and the input from last round (initial input or any received messages). The output gates correspond to the local memory and the messages the machine sends out after the local computations are concluded. Overall, these are $\bigO(S \log n)$ input and output gates. Additionally, each $C^{ij}$ has a dedicated result gate, which is 1 if and only if the machine $i$ decides yes in round $j$.
    Since $T^{ij}$ is \P-bounded, it conducts at most $\poly(S\log n)$ local computation steps. Therefore we can apply the standard Turing machine to circuit transformation (e.g., \cite{DBLP:books/daglib/0086373,DBLP:books/daglib/0023084}), and convert each $T^{ij}$ into an $\Lclass$-uniform circuit $C^{ij}$ of size and depth $\poly(S\log n)$.    
    
    It remains to determine how to connect the $T^{ij}$ into a single circuit that simulates $\mathcal A$. We first take care of connecting the result wires of all $C^{ij}$ via a balanced tree of OR-gates with fan-in 2 (adding at most depth $\log (tM)$ and size $tM$), which is the dedicated output of the whole circuit $C_n$. Since $\mathcal A$ solves the decision problem $\Pi$, at least one machine $i$ will output yes in some round $j \leq t$ on instance $w \in \{0,1\}^n$ iff $w \in \Pi$, consequently the result wire of circuit $C^{ij}$ is 1, iff $w \in \Pi$.
    
    To make the circuit functional, the outputs of $C^{ij}$ must be connected appropriately to the inputs of $C^{i, j+1}$ (to forward the local memory of machine $i$) and into all $C^{i', j+1}$ for all machines $i'$ that machine $i$ sends a message to in round $j$.
    Then the overall circuit $C_n$ is formed by the $C^{ij}$ arranged in $t$ layers.     
    Considering the size and number of circuits $C^{ij}$ and the number of layers, so far $C_n$ has overall size $\bigO(tM\poly(S\log n))$ and depth $\bigO(t\cdot  \poly(S\log n))$ (the OR-tree connecting the result wires makes no difference asymptotically).
    
    However, for connecting the layers we have to handle a crucial issue: while the communication among \MPC machines can depend on the input itself, the wiring pattern between the $C^{ij}$ must be oblivious to the input. Therefore we first apply Lemma \ref{lem:mpc_oblivious} to turn $\mathcal A$ into an equivalent input-oblivious algorithm that takes only slightly longer ($t \cdot \polylog(MS) = t \cdot \polylog(n)$ rounds, exploiting $MS  = \poly(n)$), which ensures that our message exchange pattern among machines can be output by a log-space Turing machine.      

    We argue that the resulting circuit is $\Lclass$-uniform for the parameter range $M,S,t \in \bigO(\poly(n))$.     
    The Turing machine first serially outputs all $C^{ij}$, reusing the log-space working memory (using only an additional $\bigO(\log(MSt)) = \bigO(\log n)$ length counter to keep track of the number of circuits that we output so far). As argued before, the interconnections between the $C^{ij}$ can also be output by a log-space bounded Turing machine since we can make algorithm $\mathcal A$ input-oblivious. Hence $\{C_n\}$ is \Lclass-uniform for $S,M,t \in \bigO(\poly(n))$.
\end{proof}

Next we show that it is also possible to extract a circuit from a \NCC algorithm. Since the input in the \NCC model is tied to an input graph, we specifically require that $\Pi \subseteq \{0,1\}^*$ encodes a graph decision problem (e.g., triangle detection). We say that a $\NCC(C)$-algorithm $\mathcal A$ decides $\Pi$ in $t$ rounds if, for any $w \in \{0,1\}^m$ encoding an instance of the graph decision problem on an $n$-node graph, at least one node decides yes within $t$ rounds if and only if $w \in \Pi$.

We extend the notion of $P$-bounded to \NCC-algorithms: in every round, each node performs at most $\poly(k)$ local steps, where $k$ is the size of the cumulative information it has seen so far (its initial input plus all messages received). As both the input size and the number of rounds are typically at most polynomial in $n$, $P$-boundedness implies that every node’s local work stays polynomial in $n$. Strictly speaking, there is no function $f(n)$ limiting local steps of \NCC nodes, but it is standard not to rely on this when designing algorithms and necessary to allow for polynomial-size circuits.

\begin{lemma}
\label{lem:ncc_to_circuit}
Let $\Pi\subseteq\{0,1\}^*$ be a decision problem on graphs, and let $\mathcal A$ be a $\Pclass$-bounded $\NCCx{C}$ algorithm that decides $\Pi$ in $r$ rounds. Let $n$ be the number of nodes of a problem instance $w \in \{0,1\}^m$ and let $I \in \bigO( \poly(n))$ be the maximum number of words in a node’s initial input. If $C \in \bigO( \poly(n))$ then there exists an $\Lclass$-uniform circuit family $\{C_m\}$ deciding $\Pi$ with fan-out at most 2, depth at most $\xi$ and size at most $n\cdot\xi$ for \smash{$\xi = \poly\bigl(r(C+I)\log n\bigr)$}.
\end{lemma}

\begin{proof} By $\Pclass$-boundedness, the local computation at node $i$ in round $j$ can be modeled as a polynomial-time Turing machine $T^{ij}$ that runs for at most $\poly(r(C+I)\log n)$ steps on the data consisting of the node’s initial input ($\leq I$ words) and all messages received by the end of round $j{-}1$ ($\leq rC$ words). This bound still holds if nodes do not store any intermediate results across rounds and instead recompute them from scratch each round: the recomputation adds at most a factor $r$, which is absorbed by $\poly(r(C+I)\log n)$.
We again apply the Turing machine to circuit transformation (\cite{DBLP:books/daglib/0086373,DBLP:books/daglib/0023084}), each $T^{ij}$ can be converted into an $\Lclass$-uniform circuit $C^{ij}$ of size and depth $\poly(r(C+I)\log n)$. 
Each $C^{ij}$ has a decision wire that is $1$ iff node $i$ outputs ``yes'' in round $j$. We connect all decision wires via a balanced binary OR tree to obtain the circuit’s final output.

We now show how to inter-connect the $C^{ij}$ to obtain $C_m$. Circuit $C^{ij}$ has input gates for the node’s problem-input bits and for all message bits received up to the end of round $j-1$, and output wires for all message bits sent in round $j$. We bundle message bits into \emph{ports} of $\Theta(\log n)$ wires (one machine word). Each $C^{ij}$ has exactly $C$ output ports (capacity $C$ words per round), and $Cr$ input ports (capacity $C$ words per round over at most $r$ rounds). When node $i$ sends a word in round $j$ to some node $i'$, we connect the corresponding output port of $C^{ij}$ to the appropriate input port of all $C^{i'j'}$, $j' > j$. Overall the number of wires incident to $C^{ij}$ is bounded by $\poly(r(C+I)\log n)$.

However, the last step produces gates with a high fan-out, therefore we bound the fan-out by $2$ using copy trees.  
Whenever a gate feeds into more than two successors, we replace this high-fan-out with a balanced binary tree of “copy” gates (e.g., $x\land 1$), such that each internal copy gate has at most two outputs. This increases the size by only a constant factor (each copy gate can be charged to a consumer gate and each consumer gate is charged at most twice by a copy gate), and multiplies the depth by at most a logarithmic factor in the circuit size. Thus the overall size is $n\cdot\poly(r(C+I)\log n)$ and the overall depth remains $\poly(r(C+I)\log n)$.

It remains to argue the $\Lclass$-uniformity of the resulting circuit $C_m$. The Turing machine that outputs $C_m$ sequentially runs the log-space Turing machines to print the $C^{ij}$ (each time reusing the log-space working memory and keeping track of the current circuit $C^{ij}$ via a counter of size $\log(nr) \in \bigO(\log n)$).
The connection of the output wires via an OR tree and the reduction of high fan-out gates using trees of copy gates can be described by basic arithmetic operations, i.e., can be accomplished by a log-space bounded Turing machine.

The remaining issue is that the interconnections between the $C^{ij}$ are determined by the communication between machines during the execution of $\mathcal A$ which may depend on the problem input. We construct an equivalent input-oblivious algorithm $\mathcal A'$ in advance of this proof, by applying Corollary \ref{cor:ncc_oblivious}. Thus we obtain the property that there exists a $\log$-space bounded Turing machine that outputs the connections among the $C^{ij}$, as required for $\Lclass$-uniformity. This increases the running running time $r$ only to $r' \in \bigO(r \cdot \polylog(nC))$ which does not change our result asymptotically.
\end{proof}

\subsection{Conditional Impossibility} 

We now relate the bandwidth requirement of Theorem~\ref{thm:ncc_sim_mpc} to circuit complexity.
Let $\NCy{near-lin}$ denote the subclass of $\NC$ decidable by uniform circuits of size $n\cdot\polylog(n)$ and polylogarithmic depth.
Under the local-computation restrictions specified below, the bandwidth reduction considered here would imply $\NCy{near-lin}=\NC$.
Thus every decision problem in $\NC$ would admit uniform, near-linear-size circuits while retaining polylogarithmic depth, a major step towards resolving \smash{$\NCy{lin}\stackrel{?}{=}\NC$}.

Consequently, assuming $\NCy{near-lin}\neq\NC$, our positive simulation is nearly tight in its bandwidth requirement under these computational restrictions.
This conditional impossibility holds even for unlabeled input graphs of constant arboricity and polylogarithmic capacity $C$.

\begin{theorem}    
    \label{thm:sim_mpc_in_ncc_impossible3}
    For any constant $\varepsilon>0$ and any fixed $t=\polylog(n)$, on the class $\mathcal D$ of decision problems on unlabeled, $n$-node graphs with $MS\ge n^{1+\varepsilon}$ with $\Pclass$-bounded machines and nodes, we have $\MPCx{M,S}\nleq_{\mathcal D,t}\NCCx{MS/n^{1+\varepsilon}}$ unless $\NCy{near\text{-}lin}=\NC$.
\end{theorem}

\begin{proof}    
    In this proof we distinguish two parameters which we have previously identified with $n$. Let the parameter $n$ denote the number of nodes in our graph problem, and in particular, the number of computational nodes in the $\NCC$-model. The parameter $m$ denotes the length of the input that is handled by a circuit $C_m$.
    We aim to show the contrapositive, i.e., $\NCy{near-lin} = \NC$ under the assumption $\MPCx{M,S} \leq_{\mathcal D,t} \NCCx{MS/n^{1+\varepsilon}}$.

    Consider a problem $\Pi \subseteq \{0,1\}^*$ with $\Pi \in \NC$. Let $\{C_m\}$ be the $\Lclass$-uniform family
    of circuits that decide $\Pi$ and have size $s \in \Theta(m^{c})$ for some constant $c \geq 1 + \varepsilon$, depth $d = \polylog(m)$ and fan-out at most 2. We first turn $\Pi$ into a graph decision problem on a subset of constant degree, connected graphs. For some instance $w = w_1 \cdots w_m$ of $\Pi$ let $G_w=(V,E)$ with $n := 3m^{c-\eps} \geq 3m$ nodes. $V$ contains nodes $v_{1}, \dots, v_{n}$ connected on a path in that order, where the subscript $i$ corresponds to node ID of $v_i$ (w.l.o.g.). The nodes $v_1, \dots, v_{3m}$ are dedicated nodes that encode the problem input with additional edges: for the $j$-th input bit $w_j$ of instance $w$ we have $\{v_{j},v_{j+m}\} \in E$ if $w_j = 1$ and $\{v_{j},v_{j+2m}\} \in E$ if $w_j = 0$. We define a corresponding graph problem $\Pi'$ on the set of input graphs $G_w$ with $w \in \{0,1\}^m$. A distributed graph algorithm (\MPC or \NCC) decides $\Pi'$, when at least one machine/node decides yes on $G_w$, if and only if $w \in \Pi$.
    
    We construct a $\MPCx{M,S}$ algorithm $\mathcal A$ from $C_m$ solving $\Pi'$ on $G_w, w \in \{0,1\}^m$ with $M := s/n^\varepsilon, S = 2n^\varepsilon$ in $r = \polylog(n)$ rounds. We define a canonical assignment of gates to machine IDs from bottom to top layer and from left to right within each layer with $n^\varepsilon$ assigned gates per machine. Then the machines construct the output-wires of their gates controlled by other machines. This can be done locally on $\bigO(\log n)$ bits memory exploiting the $\Lclass$-uniformity. Then we sort the input edges of $G_w$ by smaller node ID of the endpoints in constant rounds (Lemma \ref{lem:mpc_sort}). This sorting step (together with our layer-wise assignment of gates to machines) ensures that the machine that simulates the input gate $j$ (which has fan-in 1) can locally read the corresponding input bit $w_j$ from the edge $\{v_j,v_\ell\}$ (with $w_j = 1$ iff $\ell = j+m$).
    Starting from layer 0, each gate is simulated in the round corresponding to its depth in the circuit and the outputs of each simulated gate are passed to the machine simulating the input gate via a message in the \MPC model. Note that $2n^{\eps}$ total messages are sufficient since the fan-out is at most 2. Overall, the $\MPC$-algorithm $\mathcal A$ takes $\bigO(1)$ rounds for sorting the input and $d = \polylog(n)$ rounds for simulating a layer of $C_m$ so $r = \polylog(n)$ rounds overall. Furthermore, all steps are polynomial in the local memory.
    Now apply the initial simulation assumption on $\mathcal A$ to obtain a $\bigO(r \cdot t)$ round $\NCCx{C}$ algorithm $\mathcal A'$ with $C \in \bigO\big(\frac{MS}{n^{1+\varepsilon}}\big)$ which satisfies
    \begin{align*}
        \bigO\Big(\frac{MS}{n^{1+\varepsilon}}\Big) = & \hspace{2mm}  \bigO\Big(\frac{s}{n^{1+\eps}}\Big) = \bigO\Big(\frac{m^c}{(3m)^{(c-\eps)(1+\eps)}}\Big) \\ 
        \subseteq & \hspace{2mm} \bigO\Big(\frac{m^c}{m^{(c-\eps)(1+\eps)}}\Big) = \bigO\Big(\frac{m^c}{m^c m^{c\eps-\eps-\eps^2}}\Big) \\
        = & \hspace{2mm} \bigO\Big(\frac{1}{m^{\eps(c-1-\eps)}}\Big) \stackrel{c \geq 1+\eps}{\subseteq} \bigO(1).
    \end{align*}
    In summary the $\NCCx{C}$ algorithm $\mathcal A'$ solves $\Pi'$ in $G$ with node capacity $C=\bigO(1)$ (see above, but the following argument works also for $C = \polylog(n)$), with input size $I=\bigO(1)$ due to the constant degree input graph, and with round complexity $r \cdot t = \polylog(n)$. 
    We turn $\mathcal A'$ back into a circuit $C'_{n}$ using Lemma \ref{lem:ncc_to_circuit} (exploiting that the nodes in the \NCC model are $\Pclass$-bounded). The new circuit $C'_{n}$ has size $\bigO(n \cdot \log^{\ell} n) \subseteq \bigO(m^{c-\varepsilon} \cdot \log^{\ell} m)$ and depth $\log^{\ell} m$ for some constant $\ell \in \mathbb N$. We have thus decreased the size of the circuit by a factor almost $m^\eps$ with depth increasing only by a $\polylog(m)$ factor. $C'_n$ takes $\bigO(n)$ input bits of the encoding of $G_w$, where $n \geq 3m$. However, since $w$ uniquely describes $G_w$, $C'_n$ can be transformed into $C'_m$ that requires only the $m$ input bits of $w$ by adding an initial adapter layer that transforms $w$ into the input format of $C'_n$ (and does not change the size asymptotically).
    
    We repeat this approach for at most $c/\eps$ (constant!) times until the size of the resulting circuit is only $m^{c} \cdot \polylog(m)$ for some $c < 1 + \eps$. If $c = 1$, then we are done. However, if $1 < c < 1 + \eps$ then we cannot apply the size reduction above as we assumed minimum size $m^c$ for $c \geq 1+\eps$. However, for $c < 1+\eps$, we can simply pad the circuit with additional gates that have no functionality to formally satisfy $c = 1+\eps$, and obtain circuit size $m \cdot \polylog( m)$ in one iteration of the procedure above.
    In conclusion, we have shown that the assumed simulation result allows us to iteratively transform any circuit in $\NC$ into a near-linear sized circuit in $\NCy{near-lin}$ while maintaining polylogarithmic depth.
\end{proof}

\subsection{Hardness of Proving Unconditional Impossibility} Superficially, it may seem self-evident that we cannot simulate any given \MPC algorithm in \NCC with only polylog overhead for a decision problem if we have less global bandwidth available in the latter than in the former. However, below the surface this question is much more fundamental. We demonstrate that if there is a $\Pclass$-bounded $\MPC(M,S)$ algorithm $\mathcal A$ that solves some decision problem with huge global bandwidth $MS$ (which is only bounded by some arbitrarily large polynomial in $n$) and proving unconditional impossibility of simulating this algorithm $\mathcal A$ with any polylog overhead in $\NCC(C)$ with some potentially very small $C \ll MS/n$ then this separates 
\smash{$\NC_{\text{near-lin}}$} from $\NC$ (and thus also $\NC_{\textnormal{lin}} \neq \NC$).

\begin{theorem}    
    \label{thm:sim_mpc_in_ncc_impossible4}
    Let $\mathcal A$ be a $\polylog(n)$-round, $\Pclass$-bounded $\MPCx{M,S}$-algorithm that solves a decision problem $\Pi$ on labeled, $n$-node graphs $G$ with maximum degree $\Delta$ with $S = \polylog(n)$ and with $M$ only bounded by an arbitrarily large, fixed polynomial in $n$.
    If one can show, for such an algorithm $\mathcal A$, that no $\Pclass$-bounded $\NCC(\Delta)$ algorithm can simulate $\mathcal A$ with overhead $t=\polylog(n)$, then $\NC_{\textnormal{near-lin}}\neq\NC$.
\end{theorem}

\begin{proof}
    We prove the contrapositive: assuming $\NC_{\textnormal{near-lin}} = \NC$, i.e., any decision problem in $\NC$ has a near-linear size, polylog-depth $\Lclass$-uniform circuit family $\{C_n\}$, then we can in fact simulate the $\MPCx{M,S}$-algorithm $\mathcal A$ on $G$ in $\NCCx{\Delta}$ with at most $\polylog (n)$ overhead. 
    The graph decision problem $\Pi$ on labeled $n$-node graphs $G$ with maximum degree $\Delta$ can be interpreted as a decision problem on $\{0,1\}^*$ where instance $G$ is encoded by $w_G \in \{0,1\}^m$ with $m \in \bigO(n\Delta \log n)$ input bits, corresponding to a constant number of machine words per node and edge (see Definition \ref{def:problem}).
    
    The first step is to apply Lemma \ref{lem:mpc_to_circuit} to extract an $\Lclass$-uniform circuit family $\{C_m\}$ from $\mathcal A$ that solves $\Pi$ on such inputs $w_G$. With our choice of parameters this circuit has depth $d = \polylog(m) = \polylog(n)$ but can be large in size as our only bound on $M$ is some (possibly large) polynomial in $n$. However, since we assumed $\NC_{\text{near-lin}} = \NC$, there exists an $\Lclass$-uniform circuit family $\{C'_m\}$ with size $s' = \bigO\big(m \cdot \polylog(m)\big) = \bigO\big(n \Delta \cdot \polylog (n)\big)$ and depth $d' = \polylog(m) = \polylog(n)$ such that $C'_m$ computes the same output as $C_m$.
    
    Then $C'_m$ can be simulated in the $\NCC(\Delta)$ model as follows. First, each node uses the generator Turing machine to locally construct the complete circuit $C'_m$ based only on the knowledge of the problem size $m$ (which can be aggregated beforehand in $\bigO(\log n)$ steps). We use a canonical assignment of the $s'$ gates to the $n$ nodes, layer by layer from left to right, at most $\bigO\big(\Delta \polylog (n)\big)$ gates per node, such that each node has to simulate at most $\Delta/2$ gates per layer. If necessary we stretch a ``large'' layer with $\Omega\big(n\Delta \cdot \polylog (n)\big)$ gates to $\polylog(n)$ layers of $n\lfloor\Delta/2\rfloor$ gates each, which increases the depth only by a $\polylog(n)$ factor.    
    The nodes then simulate their assigned gates, one layer each round. The bits of the problem input $w_G$ can be constructed locally by each node from their respective node label and the labels of incident edges. Note that each node does only $\Pclass$ computations on its local data.
    
    Each round, each node requires at most $\Delta$ input values to simulate the output of its at most $\Delta/2$ gates with fan-in $\leq 2$. These inputs can be forwarded to that node by the corresponding nodes from the previous layer via $\NCC(\Delta)$ messages.
    Doing the simulation for all layers takes at most $d'\polylog (n) = \polylog(n)$ steps. Overall we have simulated $C'_m$ and thus $\mathcal A$ in the $\NCCx{\Delta}$ model in $\polylog(n)$ rounds.
\end{proof}

The theorem above also gives clues under which conditions an impossibility result for simulation of \MPC decision algorithms in \NCC might still be attainable, without conditioning on \smash{$\NClin \stackrel{?}{=} \NC$} or a similar open problem. In particular, the proof still leaves room for such an impossibility result if we deal with larger machines $S \in \omega(\polylog (n))$ that are not $\Pclass$-bounded. For instance, it seems plausible that solving some ``hard'' sub-problem of a decision problem on a single large machine cannot be simulated efficiently when \NCC nodes only have much smaller bandwidth $C \ll S$, since we cannot efficiently aggregate data of size $S$ on a single \NCC node.

\section{Impossibility of Simulating \NCCz in \MPC}\label{sec:nccz_in_mpc_impossibility}

We now consider to what extent \NCC algorithms can be simulated in the \MPC model. Unfortunately, for simulations in this direction we have to be rather pessimistic. We begin by showing that simulating even the weak model $\NCCz$ inside $\MPC$ is impossible in general. We construct a graph problem based on the problem of comparing regular expressions with squaring operator which has exponential space-complexity and is unsolvable in the \MPC model when only polynomial memory is available.

\subsection{General Impossibility of Simulating \NCCz in \MPC}

\subparagraph*{Preliminaries} For a finite alphabet $\Sigma$, the set of regular expressions with squaring $\textnormal{RSQ}(\Sigma)$ is defined recursively as follows. Any singleton $\{r\}$ with $r \in \Sigma$ is in $\textnormal{RSQ}(\Sigma)$. For $r,r_1,r_2 \in \textnormal{RSQ}(\Sigma)$ the usual expressions $r_1 \cup r_2, r_1 \cdot r_2, r^*$ are in $\textnormal{RSQ}(\Sigma)$ but additionally the square of an expression $r^2 := r \cdot r$ is in $\textnormal{RSQ}(\Sigma)$ as well. Let $L(r)$ be the language expressed by $r \in \textnormal{RSQ}(\Sigma)$. We say that a Turing machine recognizes $\textnormal{RSQ}(\Sigma)$ if for any $r_1, r_2 \in \textnormal{RSQ}(\Sigma)$ it solves the question \smash{$L(r_1) \stackrel{?}{=} L(r_2)$}.
Note that $L(r)$ for $r \in \textnormal{RSQ}(\Sigma)$ still defines a regular language as we can always write $r \cdot r$ instead of $r^2$. However using $\textnormal{RSQ}(\Sigma)$ allows to express a regular language much more succinctly than when only the basic operations $\cup, \cdot, {}^*$ are available, for instance, the single word language \smash{$\{a^{2^k}\}, k \in \mathbb N$} requires $2^k$ symbols to express using the basic operators but only $\bigO(k)$ with the squaring operator. It was shown in \cite{MeyerStockmeyer1972} that recognizing a regular expression becomes $\EXPSPACE$ if the squaring operator is allowed.

\begin{lemma}[cf.\ \cite{MeyerStockmeyer1972}]
    \label{lem:regex_square}
     There is a finite alphabet $\Sigma$ such that if any Turing machine $T$ recognizes $\textnormal{RSQ}(\Sigma)$, then there is a constant $c > 1$ such that $T$ requires space $c^n$ on some input of length $n$ for infinitely many $n$.
\end{lemma}

We start with an auxiliary statement that any problem on connected graphs can be solved by a deterministic $\NCCz$ algorithm in a polynomial number of rounds independently of labels or node identifiers. 

\begin{lemma}
    \label{lem:nccz_solvability}
    Any problem $\Pi$ on a labeled, connected $n$-node graph is $\poly(n)$-round solvable in $\NCCz(1)$.
\end{lemma}

\begin{proof}
    First we elect a root and form a spanning tree. Each node maintains the smallest ID it has seen so far and, in each round, forwards that value to exactly one neighbor in round-robin order. The global minimum ID propagates to all nodes within $\bigO(n|E|)$ rounds. Each node designates as its parent the neighbor from whom it first learned the minimum ID, yielding a spanning tree rooted at the node with minimum ID.

    Next we conduct a converge-cast of the whole graph in the tree. In the converge-phase, every non-root node sends the edges (and labels) that it initially knows or learned of to its parent, one word per round. Each edge is described by $\bigO(1)$ words, and any word travels at most $n-1$ hops, so all data reach the root in $\bigO(n|E|)$ rounds.
    Then the root broadcasts the collected graph down the tree in $\bigO(n|E|)$ rounds.    
    Then, every node knows the entire input graph and can solve $\Pi$ locally and output its part of the solution. Hence $\Pi$ is solvable in $\poly(n)$ rounds in $\NCCz(1)$.
\end{proof}

In the following theorem we transform the problem of recognizing regular expressions with squaring operator into a graph problem for which we show two things. First, we can solve this graph problem in the $\NCCz(C)$ model, even with tiny bandwidth of $C=1$ word per round. Second, if we want to do the same in the $\MPC(M,S)$ model, we would require huge total space $MS$ otherwise we would be in violation of the Lemma above via a reduction from the graph problem in the $\MPC(M,S)$ model back to the recognition problem on Turing machines.

\begin{theorem}    \label{thm:sim_nccz_in_mpc_infeasible}
    There is an unlabeled graph problem $\Pi$ on a connected graph with $n$ nodes that is solvable in the $\NCCzx{1}$ model (in $\poly(n)$ rounds) and there exists $\varepsilon > 0$ such that $\Pi$ is unsolvable in $\MPCx{M,S}$ for $MS \leq (1+\varepsilon)^n$ in $r$ rounds for any finite $r$.
\end{theorem}

\begin{proof}
    Let $\Sigma$ be a finite alphabet, and let $(r_1,r_2)$ with $r_1,r_2 \in \text{RSQ}(\Sigma)$ be an instance of the problem \smash{$L(r_1) \stackrel{?}{=} L(r_2)$}, and let $n := |\langle r_1,r_2\rangle| + 1$ be the bit length of a unique binary encoding applied to $(r_1,r_2)$. Build a connected, unlabeled graph $G=(V,E)$ with $V=\{v_0,\dots,v_{2n+1}\}$ as follows. For each $i\in\{0,\dots,2n-1\}$ include edge $\{v_i,v_{i+1}\}$ to create a path and ensure connectivity. We add the triangle among the nodes $v_{2n-1}, v_{2n}, v_{2n+1}$ to mark the end of the path. Furthermore, we add chords $e_i := \{v_i,v_{n+i}\} \in E$ iff the $i$-th bit $\langle r_1,r_2\rangle_i$ of the encoding equals $1$ (which does not create additional triangles).
    Then all bits of $\langle r_1,r_2\rangle$ are uniquely recoverable from the chords $e_i$ in $G$ whose order is given by the path and with the triangle at the end.
    Define the graph decision problem $\Pi$ so that on input $G$ the output is ``yes'' iff $G$ encodes some instance $(r_1,r_2)$ as defined above and \smash{$L(r_1) = L(r_2)$}. We remark that the encoding is agnostic of node IDs. By Lemma \ref{lem:nccz_solvability}, we can solve the problem in $\poly(n)$ rounds in the $\NCCzx{1}$ model.
    
    Now suppose, toward a contradiction, that $\Pi$ is solvable in $\MPCx{M,S}$ with total memory $MS \le (1+\varepsilon)^n$ for a sufficiently small $\varepsilon>0$. Each machine word stores $\Theta(\log n)$ bits, so the system has $\Theta(MS\log n)$ bits memory total. We construct a Turing machine $T$ that on input $\langle r_1,r_2 \rangle$ first constructs $G$ then simulates the $\MPCx{M,S}$ model (with $MS \geq |E|$) on $G$ and decides the question \smash{$L(r_1) \stackrel{?}{=} L(r_2)$} based on whether some machine outputs yes after $r$ rounds or not.
    $T$ requires only the $\Theta(MS\log n)$ bits of total space for the construction of $G$ and the simulation. Let $c := 1+2\varepsilon$. Then for all large enough such $n$, $MS\log n \leq (1+\varepsilon)^n\log n < (1+2\varepsilon)^n = c^n$. Note that the constant $c = 1+2\varepsilon$ can assume any value $>1$, since $\varepsilon$ is arbitrarily close to $0$.    
    However, by Lemma~\ref{lem:regex_square} there exists a constant $c_0>1$ such that $T$ uses at least $c_0^n$ bits of space for infinitely many $n$, a contradiction.
\end{proof}

The theorem above implies the following simulation impossibility result.

\begin{corollary}
    There exists a constant $c>1$ such that $\NCCzx{1} \not\leq_t \MPCx{M,S}$ on $n$-node graph problems for any $MS \leq c^n$ and any $t \in \mathbb N$.
\end{corollary}

\subsection{Impossibility of Efficiently Simulating \NCCz in \MPC}

The sublinear-memory $\MPC$ model is limited in a sharper sense than failing on tasks that exceed its total memory. For instance, a single machine cannot, in general, store the full neighborhood of a node $v$ or the corresponding edge-labels, since the adjacency set $N(v)\subseteq V\setminus\{v\}$ can be any of $2^{n-1}$ subsets, so specifying $N(v)$ requires $n-1$ bits in the worst case for which the $S \in o(n)$ bits available in the sublinear regime are insufficient. By contrast, in the $\NCCz$ model the node corresponding to $v$ is given its adjacency set and associated labels by definition. Hence any $\NCCz$ algorithm where a node computes an arbitrary function $h$ on its neighborhood in 0 rounds can, in general, not be simulated without communication among machines.

Even when evaluating $h$ requires only polynomial work and little additional space beyond its input, there exists such a function that is hard to compute collaboratively on multiple machines in the random-oracle model \cite{DBLP:conf/spaa/ChungHS20}. Its computation is hard to break into subproblems that can be solved concurrently and requires $\Omega(n/S)$ communication rounds in sublinear $\MPC(M,S)$, provided that each machine makes only polynomially many oracle queries per round.
This lower bound even holds if machines may perform arbitrary finite
local computations using unrestricted temporary workspace.
We only require that each local computation phase starts with at most
$\bigO(S)$ words of information, including retained state and received
messages (besides respecting the stated communication and oracle-query bounds).

\subparagraph*{Preliminaries}
In the random-oracle (RO) model, algorithms have query access to a uniformly random function. For an integer $\ell\geq 1$, we use an oracle $H:\{0,1\}^{2\ell}\to\{0,1\}^{2\ell}$, sampled once and kept fixed throughout the computation. A query specifies an argument $y$ and returns $H(y)$. Values at distinct arguments are independent uniform strings, while repeated queries to the same argument always return the same value.

We extend our distributed models by giving every node or machine access to the same oracle $H$. During each local computation phase, including the phase before the first communication round, each node or machine may make at most $q$ adaptive queries, meaning that later queries may depend on earlier answers. Oracle queries are local operations and do not consume communication bandwidth or rounds. The existing memory bounds apply to query arguments, answers, and all other locally stored data, the oracle itself is an external resource. The admissible range of $q$ is specified in the lower-bound lemma below.

We use the following simplified version of~\cite{DBLP:conf/spaa/ChungHS20}. For an integer $b\geq 1$, define
$h^H_{b,\ell}:\{0,1\}^{b\ell}\to\{0,1\}^{2\ell}$.
Its input is a string $X=X_0\circ X_1\circ\ldots\circ X_{b-1}$
consisting of $b$ blocks $X_i\in\{0,1\}^{\ell}$, where $\circ$ denotes
bitstring concatenation.
Starting from $a_0=0^\ell$, for $i=0,\ldots,b-1$ compute
\[
    z_{i+1} = H(X_i\circ a_i), \qquad
    a_{i+1} = \operatorname{prefix}_{\ell}(z_{i+1}),
\]
where $\operatorname{prefix}_{\ell}$ selects the first $\ell$ bits of
its argument.
The output is $h^H_{b,\ell}(X)=z_b$.
We write $h$ when the parameters and the oracle are clear from context.
The following lemma states the property of $h$ that we need.

\begin{lemma}[cf.\ \cite{DBLP:conf/spaa/ChungHS20}, adaptation of Lemma A.2]
\label{lem:chs_hard_function}
Let $b,M,S,q,\ell$ depend on $n$ such that
$\log(Mb) \in o(\ell)$, $\ell\leq S\log n \in o(b\ell)$, and $q\leq 2^{\ell/4}$.
Consider an oracle $\MPC(M,S)$ algorithm that makes at most $q$ oracle queries per machine in each local computation phase.
If the algorithm computes $h^H_{b,\ell}(X)$ with probability at least $2/3$, then it requires $\Omega(b\ell/(S\log n))$ communication rounds.
Here, the input $X$ is initially partitioned among the machines, and the probability is taken over independent uniformly random choices of $H$ and $X$, as well as any internal randomness of the algorithm.
\end{lemma}

The intuition is that every step requires both the preceding state $a_i$ and the next input block $X_i$. The state is determined by the preceding oracle answer, making it difficult to compute later steps in advance. Meanwhile, at the beginning of each local computation phase, a machine
has at most $S$ words of stored information, or $\bigO(S\log n)$ bits,
which can hold only $\bigO(S\log n/\ell)$ complete input blocks. Since a newly computed state cannot be passed to another machine during the same local computation phase, this suggests that the chain can advance by only $\bigO(S\log n/\ell)$ steps per round, leading to the lower bound. A single Turing machine with access to the entire input can instead follow the whole chain locally.

The main difficulty is to make this intuition hold even when machines store arbitrary encodings of the input and make speculative oracle queries. The compression argument of \cite{DBLP:conf/spaa/ChungHS20} shows that a small local state cannot, with significant probability, produce queries containing substantially more input blocks than its size permits. Otherwise, these queries would allow too many inputs to be encoded with too few bits. Together with the unpredictability of the intermediate states, this bounds the progress per round even when many machines collaborate.

\begin{theorem}
\label{thm:sim_nccz_in_mpc_inefficient}
In the random-oracle model, there exists a labeled graph problem $\Pi^H$ on connected $n$-node graphs that is solvable in $0$ rounds in $\NCCz(1)$ using at most $n$ oracle queries per node.
For every $S=n^\delta$ with fixed $0<\delta<1$ and every $M=n^{\bigO(1)}$, any oracle $\MPC(M,S)$ algorithm making polynomially many oracle queries per machine in each local computation phase and succeeding on every input with probability at least $2/3$ requires $\Omega(n/S)$ communication rounds.
The probability is over the random oracle and the algorithm's internal randomness, and the lower bound already holds on stars.
\end{theorem}

\begin{proof}
Let $w=\lceil\log n\rceil$ and $\ell=w^2$, and let every edge $e$ carry an input label $(k_e,x_e)\in\{0,1\}^w\times\{0,1\}^w$, where $k_e$ is an ordering key and $x_e$ contains the input bits. Each label occupies $\bigO(1)$ words. We use the common random oracle $H:\{0,1\}^{2\ell}\to\{0,1\}^{2\ell}$.

For each vertex $v$, order its incident edges lexicographically by $(k_e,x_e)$ and let $b_v=\lfloor\deg(v)/w\rfloor$. If $b_v\geq 1$, let $X^v$ be the concatenation, in this order, of the strings $x_e$ of the first $b_vw$ incident edges. Equal labels have equal strings $x_e$, so their relative order does not affect $X^v$. Thus, $X^v$ consists of $b_v$ blocks of $\ell$ bits each. The problem $\Pi^H$ requires the output $h^H_{b_v,\ell}(X^v)$ for every vertex with $b_v\geq 1$; vertices with $b_v=0$ output $\bot$. Nonempty outputs are represented explicitly as vertex-value pairs.

In $\NCCz(1)$, every node initially knows the labels of its incident edges. It can therefore form the blocks in the prescribed order and evaluate $h$ locally, using $b_v<n$ oracle queries and no communication. This requires only $\bigO(\ell)$ bits of additional workspace beyond the input and polynomial local time, since the incident edges can be processed in lexicographic label order by repeated scans of the local input. Thus, $\Pi^H$ is solvable in $0$ communication rounds.

For the lower bound, consider a star with center $c$ and edges $e_1,\ldots,e_{n-1}$, and put $b=\lfloor(n-1)/w\rfloor$. For each $i$, let $k_{e_i}$ be the $w$-bit binary encoding of $i$. Given $X\in\{0,1\}^{b\ell}$, split it into $bw$ consecutive $w$-bit strings and assign these to $x_{e_1},\ldots,x_{e_{bw}}$, respectively. Set the remaining strings $x_e$ to $0^w$. For sufficiently large $n$, only the center has a nonempty output, which is exactly $h^H_{b,\ell}(X)$.

Consequently, an $\MPC$ algorithm for $\Pi^H$ on these stars also computes $h^H_{b,\ell}$. The star topology, endpoint IDs, and ordering keys are fixed independently of $X$ and $H$. Preparing the edge records from the distributed input strings requires no communication and incurs at most a constant factor in local memory, which can be absorbed by applying Lemma~\ref{lem:chs_hard_function} with $\bigO(S)$ in place of $S$. If the graph algorithm succeeds with probability at least $2/3$ on every input, it also does so when $X$ is uniformly random and independent of $H$.

We have $b\ell=\Theta(n\log n)$ and $\ell=\Theta(\log^2 n)$. Since $M=n^{\bigO(1)}$, we obtain $\log(Mb)=\bigO(\log n)=o(\ell)$. Moreover, $S=n^\delta$ with $0<\delta<1$ gives $\ell\leq S\log n=o(b\ell)$. Every polynomial query bound $q=n^{\bigO(1)}$ satisfies $q\leq 2^{\ell/4}$ for sufficiently large $n$. Hence all conditions of Lemma~\ref{lem:chs_hard_function} hold, yielding the round lower bound $\Omega\left(\frac{b\ell}{S\log n}\right)
    = \Omega\left(\frac{n\log n}{S\log n}\right)
    = \Omega(n/S)$.
\end{proof}

\section{\NCCzs Simulation in \MPC}
\label{sec:nccz_in_mpc}

As we have seen in the previous section, we have reason to be pessimistic about simulation results even of the weaker \NCCz model in \MPC. 
The preceding section identifies two obstacles to simulating $\NCCz$ in $\MPC$. Theorem~\ref{thm:sim_nccz_in_mpc_infeasible} rules out a general simulation when the total memory is subexponential in $n$. In the random-oracle model, Theorem~\ref{thm:sim_nccz_in_mpc_inefficient} additionally gives a labeled graph problem that is solvable in $0$ rounds in $\NCCz$ but requires $\Omega(n/S)$ rounds in $\MPC$ for $M=n^{\bigO(1)}$ and $S=n^\delta$, where $0<\delta<1$ is fixed. This lower bound persists with unrestricted temporary workspace during each local computation phase, provided that each machine makes only polynomially many oracle queries per phase, starts each phase with at most $\bigO(S)$ words of information, and respects the usual communication bounds.

The main issue is that the memory restrictions of the $\MPC(M,S)$ already prohibits its simulation in principle, for any sub-exponential total memory $MS$ (Theorem \ref{thm:sim_nccz_in_mpc_infeasible}). Furthermore, efficient simulation in the sublinear regime is impossible even when we allow machines to do local computations that are unrestricted, e.g.\ via some external accessible machine, if they are only allowed allowed to keep information fitting into their memory $S$ (assuming an ideal random oracle, see Theorem \ref{thm:sim_nccz_in_mpc_inefficient}).

In particular, our impossibility results imply that whenever we allow single nodes in the $\NCCz$ model to accumulate a large amount of information, then we can construct functions on that information that the unbounded \NCCz node can compute locally but no single machine is able to replicate this computation efficiently, because it cannot locally store the input, and no subset of machines can collaboratively compute that function efficiently, as it might be inherently non-distributive.
Our goal for this section is therefore to find a minimal set of necessary restrictions of \NCCz such that simulation in \MPC is indeed feasible to identify a common core where both models align. We make the following assumptions about \NCCz with justifications provided. 

We identify the following situations where an \NCCz node has a large local accumulation of information. (a) in a labeled graph a \NCCz node has access to the edge labels of all incident edges. (b) Even without edge labels, a high degree \NCCz node has access to a large set of neighbor identifiers. (c) Worse yet, a \NCCz node that is locally unlimited in its memory may accumulate more information than the entire $\MPC$ is able to store.
We deal with this in the following way. Situation (a) forces us to limit ourselves to unlabeled graphs. Situation (c) forces us to limit the information that an \NCCz node can use. In particular, to be able to simulate $\NCCz(C)$ with $C=n^\delta$ in the $\MPCx{M,S}$ model with comparable machine space $S=n^{O(\delta)}$, we must restrict the initial input of nodes in the \NCCz model to $C$ words. Hardest to deal with is situation (b), since the ability of \NCCz nodes to identify their neighbors in the input graph is a core property of the model. We deal with (b) by first restricting our simulation to graphs with small arboricity $n^\delta$ and by equipping them with function $f_v$ that allows them to map IDs to local port numbers (instead of giving each node $v$ the full set of neighbor IDs).

\begin{definition}\label{def:ncczs}
    The $\NCCzsx{C}$ model corresponds to a restricted version of the \NCCzx{C} model. Each node $v$ has only $C$ words of local memory and knows only its own unique ID from an ID space $\mathcal I$. Instead of knowing the IDs of its initial neighbors directly, $v$ can send a message to any neighbor via port numbers in $[1, \dots, \deg(v)]$ and has access to a function $f_v : \mathcal I \to [1, \dots, \deg(v)] \cup \{\bot\}$ that is computable within local memory and maps any ID from $\mathcal I$ to the corresponding port numbers $[1, \dots, \deg(v)]$, if it is a neighbor of $v$, and $\bot$ otherwise. Similar to the \NCCzx{C}, each node can contact any node whose ID it learned.
\end{definition}

Note that in \NCCzs any given node can still contact any of its neighbors, and ask any neighbor for their identifier via a single message which it can then store locally (subject to the memory restriction).
Arguably, the \NCCzs model represents a minimal set of restrictions we have to accept in order to make a simulation in \MPC feasible, that still maintains one of the core features in the \NCCz model: the ability of nodes to identify their neighborhood.
We give a simulation scheme for \NCCzsx{C} on graph problems in $\mathcal U_a$, defined as the set of unlabeled graphs with arboricity bound $a$ in the $\MPC(n,S)$ model for $S=\bigO(\beta a\cdot(\beta a+\log_\beta n)  +  C\beta a)$ for some tuning parameter $\beta$.

For our simulation, we require a preprocessing step (\Cref{ssec:nccz_in_mpc:preprocessing}) where we map nodes to machines one-to-one, with the goal that the local computation of each node is handled by exactly one machine. As a single machine is unable to hold all neighbor identifiers of a high degree node, we also distribute each node's neighborhood among $\bigO(\beta\cdot a)$ machines, making them responsible for the communication to these neighbors. In \Cref{ssec:nccz_in_mpc:simulation}, we describe the details of the actual simulation, i.e., how the machines responsible for a node perform the actions the node would perform during the execution of any simulated algorithm. 

\subsection{Preprocessing}\label{ssec:nccz_in_mpc:preprocessing}
We start by assigning to each node $v$ exactly one machine $M(v)$ that handles the computations performed by $v$.
Specifically, we will relabel the input graph by consistently replacing all node identifiers with machine identifiers.

\begin{lemma}\label{lem:mpc_machineids}
    Let $G=(V,E)$ be a graph which is given in the \MPCx{M, S} where $E$ is partitioned among the machines (subject to the memory bound $S \geq n^\delta$ for some constant $\delta > 0$) and $M=n$. For each $v\in V$ we can relabel the edges of $G$ incident to $v$ by replacing $v$'s identifier with a unique machine identifier $M(v)$ in $\bigO(\nicefrac{1}{\delta})$ rounds.
\end{lemma}
\begin{proof}
    We start by sorting the bidirected version of each edge $E$ by their first entry, i.e., for each edge $\{v,w\}$ (or $(v,w)$), the machine holding the edge provides data items $(v,w)$ and $(w,v)$ as input for the sorting routine from \Cref{lem:mpc_sort}. 

    Now, for each tuple $(v,w)$ that is to the right of some tuple $(x,y)$ with $v \neq x$ or is the first tuple in the ordering we create a tuple $(v,1)$. Else we create $(v,0)$ for $(v,w)$. We run a prefix sum algorithm (cf.\ \Cref{lem:mpc_prefix}) on these tuples, which delivers $(v,s_v)$ to each machine that stores a tuple $(v,w)$, where the partial sum $s_v$ is a unique number in $\{1,\dots,n\}$, which we interpret as the machine identifier $M(v) := s_v$.

    Next, we locally attach $M(v)$ to each edge to obtain data items of the form $(M(v),v,w)$, which we sort by $\min \{v,w\}\circ\max\{v,w\}$, where $\circ$ denotes concatenation. This way, the two data items $(M(v),v,w)$ and $(M(w),w,v)$ representing an edge are located at the same machine without loss of generality. We merge them, obtaining a data item $\{M(v),M(w)\}$ for each edge $\{v,w\} \in E$.
\end{proof}

Even in the \NCCzs model, it is possible for a single node's initial neighborhood to be an arbitrary subset of the set of nodes. As \MPC machines responsible for such high degree nodes cannot hold enough information to enable them to simulate their node's communication, the main part of the preprocessing is to compress and distribute the knowledge about each node's initial neighborhood, making multiple machines responsible for storing it. 
Specifically, we start by performing the forest decomposition described in \Cref{lem:mpc_forest_decomposition} (taking $\bigO(\nicefrac{(\log_\beta n)}{\delta})$ rounds in the \MPC$(M,S)$ model with $S\geq n^\delta$) and distribute the knowledge of each node's neighborhood as described in the following lemma. 

\begin{lemma}[Neighborhood Redistribution]\label{lem:neighborhood_redistribution}
    Presume we are given a forest decomposition $\mathcal{H}=\{F_1,\dots,F_d\}$ of a graph $G=(V,E)$ where $d\leq 4\beta a$ with arboricity bound $a$, i.e., $\arb(G)\leq a$, and a partition $\mathcal{H}=\mathcal{F}\cup \mathcal{F}'$ into two disjoint sets $\mathcal{F}\cap \mathcal{F}'=\emptyset$ of size $\lvert\mathcal{F}\rvert=\lvert\mathcal{F}'\rvert=\nicefrac{d}{2}$, where the forests in $\mathcal{F}$ have depth $\bigO(\log_\beta n)$ and the forests in $\mathcal{F}'$ have in-degree $\leq d$ (distributed among machines, cf.\ \Cref{lem:mpc_forest_decomposition}). Then, in $\bigO(\nicefrac{(\log_\beta n)}{\delta})$ rounds of the $\MPC(n,S)$ model with $S=\bigO(\beta a\cdot(\beta a+\log_\beta n))$, we can ensure that for every node $v\in V$ the machine $M(v)$
    \begin{enumerate}[(i)]
        \item learns the machine identifier $M(w)$ for each out-neighbor $w$ of $v$ in any forest $F\in\mathcal{H}$,
        \item learns the machine identifier $M(w)$ for each in-neighbor $w$ of $v$ in any forest $F\in\mathcal{F}'$, and
        \item learns its part of a distributed data structure allowing $M(v)$ to infer the machine identifier $M(v,w)$ that knows the machine identifier $M(w)$ for each in-neighbor $w$ that $v$ has in any forest $F\in\mathcal{F}$.
    \end{enumerate}   
\end{lemma}
\begin{proof}
    By construction, each node $v$ has one out-neighbor for each of the $d=\bigO(\beta a)$ forests. Thus, $M(v)$ can store the machine identifiers corresponding to the $\bigO((\beta a))$ out-neighbors directly. Similarly, $v$ has $d=\bigO(\beta a)$ in-neighbors in each of the $\nicefrac{d}{2}=\bigO(\beta a)$ forests of $\mathcal F'$, which $M(v)$ also stores directly. Any machine holding an edge of one of these two types directly sends a copy to the corresponding machines. Both procedures take one \MPC round. Storing the identifiers takes $\bigO((\beta a)^2)$ machine words of capacity in each machine.

    The number of in-neighbors in $\mathcal{F}$ that $v$ has can be much larger (up to $|V|-1$). Therefore, we will distribute the machine identifiers corresponding to a node's neighborhood among machines and store compressed information on where to look them up. More specifically, we assign a label $\ell$ to each machine and each forest, such that sorting the machines by this label results in a BFS-order traversal of that forest, where all children of a node are on machines with adjacent identifiers when performing a one-to-one sorting. Then, the corresponding machine only has to learn the identifier of its first child and the number of children it has to be able to access all of its children.
    
    Let $F\in\mathcal{F}$ be a forest with depth $\bigO(\log_\beta n)$. We start by having each machine $M(v)$ of node $v$ learn its in-degree by performing a group aggregation (cf.\ \Cref{lem:mpc_groupagg}) with $f=\SUM$ where $M(v)$ provides the input $(M(v),0)$ and a machine holding an edge $(M(u),M(w))$ provides the input $(M(w),1)$. We also have each machine $M(v)$ of node $v$ learn $M(w)$, where $w$ is the parent of $v$ by sending the edge $(M(v),M(w))$ to $M(v)$ (if it exists; recall that each machine has only one out-neighbor in each forest).
    
    Next, we have each machine $M(v)$ of a root node $v$ of $F$ (that did not learn any parent identifier this way) set $\ell(v)\coloneqq M(v)$ and broadcast $\ell(v)$ to its children. To this end, we perform another group aggregation with $f=\max$, where every machine $M(v)$ of a node $v$ provides the input $(M(v),\ell(v))$ if it just set its label and $(M(v),\infty)$ otherwise and a machine holding an edge $(M(u),M(w))$ provides the input $(M(w),0)$. $\infty$ serves to mark the machines that did not learn their predecessor and is greater than any machine identifier. During the group aggregation, we also track the path the messages take through the broadcast tree. This way, each child $u$ of $v$ also learns a unique number $j\in\{1,\dots,\indeg(v)\}$ corresponding to its position in the broadcast tree.
    
    We repeat the following until each machine has set its identifier\footnote{We can synchronize the completion of this phase by aggregating the number of nodes that already set its identifier. The resulting overhead is constant.}: Every machine $M(v)$ of a node $v$ that received its parent $w$'s label $\ell(w)$ last round sets $\ell(v)\coloneqq \ell(w)\circ j$ and broadcasts $\ell(v)$ to its children with the same procedure as described above.
    
    Finally, for each node $v$, we have the machine $M(v)$ create a data item $(\ell(v),M(v))$. We use a one-to-one sorting again to sort these data items first by the number of entries in  $\ell(v)$ and then lexicographically by $\ell(v)$. As a result of this sorting, all nodes of the same subtree are in one contiguous block (they share the first entry in their $\ell$-labels), and all children of same parent are in a subblock of that block (they share all but the last entries in their $\ell$-labels), i.e., this sorting yields a BFS order traversal of the forest $F$ where children of the same parent are adjacent.
    
    For each node $v$, where the last entry of $v$ is a $1$, i.e., where $v$ is a first child of its parent, $M(v)$ sends its own machine identifier to $v$'s parent's machine. This way, every parent learns the identifier of a machine that knows the identifier of its first child's machine. As it already learned its own degree earlier, it can compute the identifiers of the machines that know the identifiers of its other children, too. Specifically, for the $j$th child $w_j$ of $v$ in some forest $F\in\mathcal F$ ($2\leq j\leq \indeg(v)$), $M(w_j)$ is stored at $M(v,w_j)=M(v,w_1)+j-1$.
    
    For a single forest the loop iterates $\bigO(\log_\beta n)$ times and each iteration contains a group aggregation that takes $\bigO(\nicefrac{1}{\delta})$ time. Thus, the total computation takes $\bigO(\nicefrac{(\log_\beta n)}{\delta})$ \MPC rounds. The computation requires $\bigO(\log_\beta n)$ machine words of capacity for a single forest. Performing it for all forests in parallel requires $\bigO(\beta a\cdot\log_\beta n)$ machine words of capacity. Afterwards, each machine stores at most two machine words of data per forest to infer the positions where data on its own children is located. Furthermore, each machine stores $\bigO(\beta a)$ machine words to store other machines' neighbors.
\end{proof}

\subsection{Simulation}\label{ssec:nccz_in_mpc:simulation}

Our goal is to perform all computations that an \NCCzs node $v$ would perform during the execution of an algorithm on an unlabeled graph problem on the machine $M(v)$.
In \Cref{ssec:nccz_in_mpc:preprocessing} we described how to distribute the knowledge of the neighborhood of $v$. 
After this preprocessing, each machine $M(v)$ responsible for a node $v$ directly stores or has access to the identifier of machine $M(w)$ for every neighbor $w$ of $v$ via constant rounds of direct communication. To properly interface with the \NCCzs model, we require each machine to be able to access its neighbors via port numbers.

\begin{lemma}
    After the preprocessing of \Cref{lem:neighborhood_redistribution}, each machine $M(v)$ can assign port numbers in $\{1,\dots,\deg(v)\}$ to its neighbors, such that for each neighbor $w$, $v$ can contact $M(w)$ via $w$'s port number. 
\end{lemma}
\begin{proof}
    $M(v)$ assigns a continuous range of port numbers starting at $1$ to the neighbors it stores directly (cf.\ (i), (ii) from \Cref{lem:neighborhood_redistribution}). 
    To handle the neighbors stored using the data structure (cf.\ (iii) from \Cref{lem:neighborhood_redistribution}), $M(v)$ assigns a continuous range of port numbers to them, i.e., if the first child corresponding to an in-neighbor in forest $F\in\mathcal{F}$ has port number $p$, the $i$-th child has port number $p+i-1$. It remains to pick a suitable port number for each first child of each $F\in\mathcal{F}$ such that the ranges of port numbers do not intersect.
\end{proof}

It remains to ensure that the nodes can evaluate $f_v$, i.e., locally map node identifiers to port numbers (or $\bot$ if the identifier does not correspond to a neighbor). Note that in our simulation, we only have to evaluate $f_v$ for at most $C$ identifiers and we may conduct $O(1)$ communication rounds for the evaluation. This is due to the fact that in order to obtain a set of $C$ (fresh) identifiers, 
an \NCCzs node has to conduct at least one round of communication (see Definition \ref{def:ncczs}),

\begin{lemma}\label{lem:neighbor_identification}
    After the preprocessing of \Cref{lem:neighborhood_redistribution}, an \MPC machine $M(v)$ simulating a node $v$ can evaluate $f_v(I)$ (cf.\ Definition \ref{def:ncczs}) for a set $I$ of at most $C$ machine identifiers it stores locally.
    The procedure requires $O(\nicefrac{1}{\delta})$ rounds in the \MPCx{$n,S$} model with $S\geq C\beta a, S \geq n^\delta$.
\end{lemma}

\begin{proof}
    We start by performing a group aggregation with $f=\bigcup$ the union operation on sets, where $M(v)$ provides the input $(M(v),I)$ and every node holding an identifier from $M(v)$'s neighborhood (cf.\ \Cref{lem:neighborhood_redistribution}) provides the input $(M(v),\emptyset)$. Note that the total size of the aggregate computed this way is at most $C$ machine words. Afterwards, each participating machine can locally check, if it holds any of the identifiers in $I$. Any node that does sends them to $M(v)$ directly together with its own identifier, allowing $M(v)$ to infer the corresponding port. As each machine only holds identifiers for at most $\bigO(\beta a)$ machines responsible for simulating nodes, this procedure does not violate the \MPC's messaging limit. In total, we require $O(\nicefrac{1}{\delta})$ rounds of communication.
\end{proof}

Finally, we prove that we can deliver all messages sent during the execution of one \NCCzs round in $O(1)$ simulation rounds in the \MPC model.

\begin{lemma}[Message Delivery]\label{lem:message_delivery}
    For any \NCCzsx{C} algorithm $\mathcal{A}$ and for every node $v$, let $w_1,\dots,w_\ell$ be the nodes $v$ aims to send messages to in one \NCCzs round during the execution of $\mathcal{A}$. Assuming that we have conducted the prepossessing step $M(v)$ can send these messages to $M(w_1),\dots,M(w_\ell)$ (see Lemma \ref{lem:neighborhood_redistribution}) in three $\MPC(n,C\cdot (\beta a))$ rounds.
\end{lemma}
\begin{proof}
    $M(v)$ can directly send messages to any machines it stores the identifier of itself, independent of whether $v$ knows those machines by identifier or port, and the messaging restriction of the \NCCzs model ensures that these messages create no congestion at a single machine. Any other machine $w_i$, $v$ could send a message to, must be of its initial neighborhood, i.e., it is only known to $v$ by port number. To send a message to the corresponding machine $M(w_i)$, $M(v)$ evaluates the arbitrary port numbering it picked, identifying the correct neighbor $M(v,w_i)$ storing $M(w_i)$. $M(v)$ can then request $M(w_i)$ from $M(v,w_i)$ and send the message. As each machine stores at most one machine identifier per forest, this can cause a congestion of $O(\beta a)$.
    
    Since $v$ sends up to $\bigO(C)$ messages, $M(v)$ needs to send $\bigO(C)$ requests to auxiliary machines and each auxiliary machine receives at most $\bigO(\beta a)$ requests (one for each forest). To handle the congestion $S \in  \bigO(C\cdot \beta a)$ is sufficient. 
\end{proof}

We combine our results to obtain the following simulation result.
\begin{theorem}
\label{thm:sim_ncczs_in_mpc}
    Let $\,\mathcal{U}_a$ be a family of unlabeled graph problems with arboricity bound $a$. For the $\NCCzsx{C}$ model, we have $\NCCzsx{C}\leq_{\mathcal{U}_a}\MPC(n,S)$, 
    where $S \in \bigO(\beta a\cdot(\beta a+\log_\beta n)  +  C\beta a)$
    for a parameter $\beta$. We require a preprocessing with additive overhead $\bigO(\nicefrac{1}{\delta}\cdot \log_\beta n)$ rounds and the overhead during simulation is $\bigO(1/\delta)$.
\end{theorem}

\begin{proof}
    We start by having the MPC machines perform our preprocessing where we first replace node identifiers with machine identifiers (cf.\ \Cref{lem:mpc_machineids}) and distribute each node $v$'s neighborhood among multiple machines, making each of them available to the machine $M(v)$ corresponding to $v$ (cf.\ \Cref{lem:neighborhood_redistribution}). This way, every machine $M(v)$ simulating $v$ can locally assign an arbitrary port numbering in the range $\{1, \dots, \deg(v) \}$ for the neighborhood of $v$, which means it has the same local view as the nodes they simulate. As it further has access to $\geq C$ memory, it can perform all local computations that $v$ would perform in the \NCCzs, including checking for $\leq C$ identifiers, which of them are neighbors (cf. \Cref{lem:neighbor_identification}), and importantly it can make the same decisions on which ports and machines to send messages to that $v$ would make. Finally, \Cref{lem:message_delivery} yields that $M(v)$ can deliver any messages to neighbors according to the assigned numbers.
\end{proof}

Setting $\beta=a=C=n^\delta$ for $\delta\in(0,\nicefrac{1}{4})$, we obtain:

\begin{corollary}
\label{cor:sim_ncczs_in_mpc}
    Let  $\delta\in(0,\nicefrac{1}{4})$ and $\mathcal{U}$ be a family of unlabeled graph problems with arboricity bound $n^\delta$. For the \NCCzsx{\ensuremath{n^\delta}} model, we have    $\NCCzsx{n^\delta}\leq_{\mathcal{U}}\MPC(n,n^{4\delta})$.
    We require a preprocessing of additive $\bigO(\nicefrac{1}{\delta^2})$ rounds and the overhead during runtime is $\bigO(1/\delta)$.
\end{corollary}

Finally, we remark that we can also use less than $n$ machines in the simulation by instead increasing $S$ with the following simulation.

\begin{lemma}
    For $x\in\{1,\dots,M\}$, and $S = n^\delta$ for some constant $\delta > 0$ we have $\MPC(M,S)\leq\MPC(\lceil\nicefrac{M}{x}\rceil,x\cdot S)$
\end{lemma}
\begin{proof}
    We assume the initial states of the $\MPC(M,S)$ machines are distributed evenly among the $\MPC(\lceil\nicefrac{M}{x}\rceil,x\cdot S)$ machines without splitting the state of a single machine. Every $\MPC(\lceil\nicefrac{M}{x}\rceil,x\cdot S)$ machine $i$ creates a data item $(j,i)$ for every $\MPC(M,S)$ machine $j$ it holds. We sort these data items by their first entry, such that $\MPC(\lceil\nicefrac{M}{x}\rceil,x\cdot S)$ machine $i$ holds data items corresponding to $\MPC(M,S)$ machines $x\cdot i,\dots,x\cdot i+x-1$. We have the machines exchange their states, resulting in $i$ actually holding these machines. Each $\MPC(\lceil\nicefrac{M}{x}\rceil,x\cdot S)$ machine simulates the machines $\MPC(M,S)$ it holds. When a simulated machine $j$ sends a message to a simulated machine $j'$, it is sent to $\lfloor \nicefrac{j'}{x}\rfloor$ instead. 
\end{proof}

\clearpage
\phantomsection
\addcontentsline{toc}{section}{\refname}
\bibliographystyle{plainurl}
\bibliography{ref}

\clearpage
\appendix
\section{Subroutines and Further Preliminaries}

In this section we list the auxiliary subroutines and results that are used in this work.

\subsection{Fundamental \MPC Routines}
\label{sec:mpc_routines}

In the following, we sketch several known primitives in \MPC for sorting, aggregation and computing prefix-sums, which can then be used to decompose a graph into a number of forests proportional to its arboricity with specific properties we require later on. 

\subparagraph*{Constant-round sorting.} In the sorting problem we have $k$ sorting keys (of size one machine word) distributed arbitrarily among machines. We need to rearrange the keys such that they appear sorted in the consecutive memories of machines $1, \dots, M$.
For \smash{$\delta=\tfrac12$}, assuming the $k$ input keys are scattered randomly among the $M=S\in\Theta(k^{1/2})$ machines, a randomized \emph{sample-sort} completes in $\bigO(1)$ rounds: (1) each machine locally sorts its $k^{1/2}$ keys and emits $\bigO(1)$ evenly spaced samples; (2) the $\bigO(M)$ samples are sent to a single collector, sorted, and every $O(1)$-th item is chosen as a splitter; (3) the $\bigO(M)$ splitters are broadcast; (4) any key that falls into the $j$-th splitter interval is routed to machine $j$; (5) each machine $j$ locally sorts the received keys from interval $j$. Chernoff bounds show that every bucket contains $\Theta(k^{1/2})$ keys with high probability, so no memory constraint is violated~\cite{DBLP:conf/isaac/GoodrichSZ11,DBLP:conf/soda/KarloffSV10}. A deterministic variant replaces the random scatter with regular sampling: each machine selects $\Theta(\log k)$ evenly spaced samples, then a histogram round computes exact bucket sizes, and the resulting $\bigO(1)$-round algorithm guarantees that every bucket holds $\bigO(k^{1/2})$ keys for any input~\cite{DBLP:conf/isaac/GoodrichSZ11}.
The algorithm extends to any constant $\delta < 1/2$ by, loosely speaking,  recursively sorting the $\Theta(M)$ samples from step (2), with recursion depth $1/\delta$. In some cases, we are interested in performing a one-to-one sorting, i.e., in sorting $O(M)$ items, such that each machine receives at most one item. We achieve this, by duplicating each item $S$ times.

\begin{lemma}[\cite{DBLP:conf/isaac/GoodrichSZ11,DBLP:conf/soda/KarloffSV10}]\label{lem:mpc_sort}
Let $\delta\in(0,1)$ be a constant. In the \MPCx{M,S} model with $M\in\Theta(k^{1-\delta})$, $S\in\Theta(k^{\delta})$, $k$ keys can be sorted in $O(1/\delta)$ rounds.
\end{lemma}

\subparagraph*{Constant-round group aggregation.} Given group-value pairs $(g_i,v_i)$, $i = 1, \dots, k$ (each of size one machine word), distributed arbitrarily over the machines, each machine that initially holds $(g_i,v_i)$ should output $f_{g_i}$, where $f_g = f\big(v_j \mid g_j = g \big)$ for some associative and commutative function $f$. For instance, this problem covers counting, computing the minimum or broadcasting within each group, which allows to implement computing group-wide averages.
First add the number $m_i$ of the machine that initially holds $(g_i,v_i)$, i.e., we obtain tuples $(g_i,v_i,m_i)$. Sort these by $g_i$ using Lemma \ref{lem:mpc_sort}.
Light groups $g$, whose values are on a single machine can be aggregated within that machine. Heavy groups $g$ that span multiple machines are aggregated up an aggregation tree to the least common ancestor of all machines with values in $g$, who broadcasts the complete aggregate $f_{g}$ back down the corresponding sub-tree. The tree has degree $\Theta(S)$ and height $h=\bigO(\log_S n) = \bigO(1/\delta)$ and can be computed locally by each machine using machine labels.
The results $f_{g_i}$ are then transferred back to $m_i$ in a single step.

\begin{lemma}[\cite{DBLP:conf/isaac/GoodrichSZ11}]\label{lem:mpc_groupagg}
Let $\delta\in(0,1)$ be a constant. In the \MPCx{M,S} model with $M\in\Theta(k^{1-\delta})$, $S\in\Theta(k^{\delta})$, group aggregation can be carried out in $O(1/\delta)$ rounds.
\end{lemma}

\subparagraph*{Constant-round prefix-sum.}  We are given $k$ indexed values $(i,v_i)$ spread arbitrarily over $M$ machines (so $MS \in \Omega(k)$).  The task is to output $(i,s_i)$ with \smash{$s_i=\sum_{j \leq i}v_j$}. \cite{DBLP:conf/isaac/GoodrichSZ11} achieves this in $\bigO(1)$ rounds when $S=k^{\delta}$, $0<\delta<1$, as follows. (1) sort pairs by index in $\bigO(1/\delta)$ rounds, i.e., the sequence of pairs $(i,v_i)$ are stored as consecutive blocks on the machines; (2) each machine computes prefix sums within its block and the block total; (3) aggregate the $M$ totals up the aggregation tree such that each treenode holds the total sum of its subtrees; (4) broadcast prefix sums of the totals of subtrees down the same tree, adding the prefix sum obtained from the parent as offset; (5) at the leaves we obtain the final $(i,s_i)$. \cite{DBLP:conf/isaac/GoodrichSZ11} also proves an $\Omega(\log_S n)$ lower bound, making the algorithm round-optimal.

\begin{lemma}[\cite{DBLP:conf/isaac/GoodrichSZ11}]\label{lem:mpc_prefix}
Let $\delta\in(0,1)$ be a constant. In the \MPCx{M,S} model with $M\in\Theta(k^{1-\delta})$, $S\in\Theta(k^{\delta})$, the prefix sum of $k$ index-value pairs can be computed in $O(1/\delta)$ rounds.
\end{lemma}

\subparagraph*{Group-based prefix-sum.} We can generalize the prefix sum routine to compute prefix sums only within the same group. Here we are given $k$ tuples $(g,i,v_{g,i})$ distributed on the machines, such that the $(g,i)$ are pairwise different among tuples. We want each machine that holds some tuple $(g,i,v_{g,i})$ to output $(g,i,s_{g,i})$ with \smash{$s_{g,i}=\sum_{j \leq i} v_{g,j}$} (where we formally set $v_{g,i} = 0$, if there is no input tuple $(g,i,v_{g,i})$). We first sort the $(g,i,v_{g,i})$ by $(g,i)$, i.e., tuples are grouped by $g$ on the machines with ascending indices $i$ within groups. We then run the prefix sum algorithm within each block (essentially broadcasting partial totals up and down the broadcast tree to the least common ancestor within each block) to compute the values $(g,i,s_{g,i})$ which will be obtained by the machine that holds $(g,i,v_{g,i})$ after sorting. Finally, we can the attach the original machine number to each $(g,i,v_{g,i})$ (before sorting) and send the result $(g,i,s_{g,i})$ back to this machine.

\begin{lemma}
    \label{lem:mpc_group_prefix}
    Let $\delta\in(0,1)$ be a constant. In the \MPCx{M,S} model with $M\in\Theta(k^{1-\delta})$, $S\in\Theta(k^{\delta})$, the group based prefix sum with $k$ inputs can be computed in $O(1/\delta)$ rounds.
\end{lemma}

\subsection{\MPC Routines for Forest Decompositions}
\label{sec:mpc_routines_forest_decomp}

Throughout this subsection, assume that $G$ has at least one edge. For an edgeless graph, the orientation and the forest decomposition are both empty.

\subparagraph*{Constant round forest decomposition} The primitives above can be used to orient edges such that each node has only bounded out-degree, in the \MPC model. We apply the Nash-Williams forest decomposition framework, which operates in phases. In each phase, low degree nodes will orient their incident edges, after which they are removed from the graph. The process can be shown to terminate after relatively few phases.

\begin{lemma}[\cite{DBLP:conf/podc/BarenboimE08,NashWilliams1964}]
    \label{lem:mpc_orientation}
    An orientation of a graph $G$ with out-degree less than $2  \beta \arb(G)$ can be computed in  $\bigO\big(\tfrac{\log_\beta n}{\delta}\big)$ rounds in the \MPCx{M,S} model with $S \geq n^\delta$ and for constant $\delta \in (0,1)$ and $\beta \geq 2$.
\end{lemma}

\begin{proof}
    To achieve the bound, we follow \cite{DBLP:conf/podc/BarenboimE08} and speed up a Nash-Williams forest decomposition by raising the degree threshold for node removals. The computation in \MPC is performed in phases until all edges are directed. 
    In each phase, we first compute the degree of every node: replace every undirected edge $\{v,w\}$ with directed edges $(v,w)$ and $(w,v)$ and compute the count in the group defined by the first node entry (Lemma \ref{lem:mpc_groupagg}), resulting in triples $(u,v,\deg(u))$. Sorting by the minimum ID of each edge's incident nodes, breaking ties with the maximum identifier, edges $(u,v,\deg(u))$ and $(v,u,\deg(v))$ can be combined to $(u,v,\deg(u),\deg(v))$.
    We also compute the average degree $\Delta_{avg}$ of the remaining nodes in this phase using Lemma \ref{lem:mpc_groupagg}.
    
    As the final step of the phase, we orient all edges $\{u,v\}$ where $\deg(u)\leq \beta\cdot \Delta_{avg}$ towards $v$ and do not consider edges incident to $u$ in future phases. If both $\deg(u)\leq \beta\cdot \Delta_{avg}$ and $\deg(v)\leq \beta\cdot \Delta_{avg}$ orient $\{u,v\}$ arbitrarily.
    At most a $1/\beta$ fraction of remaining nodes can have degree larger than $\beta\cdot \Delta_{avg}$, thus each phase, only a $1/\beta$ fraction of nodes remains. Therefore all nodes are removed after $O(\log_\beta n)$ phases and each phase takes $O(1/\delta)$ rounds according to Lemmas \ref{lem:mpc_prefix} and \ref{lem:mpc_groupagg}, the stated runtime follows. Finally, we have $\Delta_{avg} < 2 \cdot \arb(G)$, which yields the out-degree bound.
\end{proof}

We show that with a slight change to the edge orientation process we can obtain a decomposition of a graph into oriented forests with the property that each forest has either relatively small depth or small in-degree.

\begin{lemma}\label{lem:mpc_forest_decomposition}
    A partitioning of $G$ into forests $F_{1}, \dots, F_{d}, F_{1}', \dots, F_{d}'$ with $d < 2 \beta \arb(G)$ such that each forest $F_j$ has depth $O(\log_\beta n)$ and each forest $F_{j}'$ has in-degree less than $2 \beta \arb(G)$, can be computed in  \smash{$\bigO\big(\tfrac{\log_\beta n}{\delta}\big)$} rounds in the \MPCx{M,S} model with $S \geq n^\delta$, for constant $\delta \in (0,1)$ and $\beta \geq 2$.
\end{lemma}

\begin{proof}
    We conduct the phases almost as described previously (proof of Lemma \ref{lem:mpc_orientation}) with a small change when we orient an edge $\{u,v\}$ with $\deg(u)\leq \beta\cdot \Delta_{avg}$ and $\deg(v)\leq \beta\cdot \Delta_{avg}$, i.e., the case where both nodes have small degree and will be removed in the same phase. Instead of orienting $\{u,v\}$ arbitrarily we now orient it towards the node with smaller ID.
    
    Consider the strict total order on the nodes of $G$ imposed by the phase in which each node was removed (where $u < v$ if $u$ was remove before $v$), with the smaller node ID as tie breaker among nodes that are removed in the same phase.
    This strict total order is an extension of the partial order given by the edge orientation, since we direct edges either towards the node that got removed in an earlier phase, or towards the one with smaller ID if both endpoints are removed in the same phase. Consequently, the orientation that we compute is acyclic.

    For each node, let us assign each out-edge uniquely to one of the directed edge sets $H_1, \dots, H_{d}$, i.e., $d < 2  \beta  \arb(G)$.
    Since the orientation of $G$ is acyclic, each $H_i$ is acyclic as well. Furthermore, in a given $H_i$, each node $v$ has at most one parent (defined by the at most one out-edge of $v$ assigned to $H_i$). Therefore, $H_i$ cannot even have an undirected cycle, because if there would be such an undirected cycle, then this cycle would be acyclic after the orientation, implying that one node on the directed cycle would have two parents, a contradiction.

    We further subdivide edges into two categories. Category a: Edges where endpoints were removed in different phases. Category b: Edges where endpoints were removed within the same phase. We partition each $H_i$ into category a edges $F_{i}$ and category b edges $F'_{i}$. Any directed path that consists exclusively of category a edges can have a length of at most the total number of phases, i.e., $\bigO(\log_\beta n)$ implying the same for the depth of $F_{i}$.    
    Any node has (undirected) degree less than $2 \beta \arb(G)$ at the time of removal (see previous proof), i.e., each node has less than $2\beta \arb(G)$ incident category b edges, which implies the same limit for the in-degree in $F'_{i}$.
\end{proof}

\subsection{Assorted Graph Properties}
\label{sec:graph_properties}

Let $G=(V,E)$ be an undirected graph. We write $H \subseteq G$ if $H$ is a (nonempty) subgraph of $G$.
The \emph{arboricity} $\arb(G)$ is the smallest integer $k$ such that the edge set $E$ is the union of $k$ edge-disjoint forests. Intuitively, the arboricity is a density parameter that counts how many forests are needed to cover the graph. For instance, $G$ has at most $\arb(G)(n-1)$ edges and thus average degree $\Delta_{avg}(G) < 2\arb(G)$.
Nash-Williams \cite{NashWilliams1964} gave an equivalent characterisation of arboricity:
\begin{align}
\label{eq:arb-def1}
\arb(G)= &
\min\{\,k\in\mathbb{N} \mid \forall H \subseteq G,\ |E(H)|\le k(|V(H)|-1)\} \\
\label{eq:arb-def2}
= & \Bigl\lceil \max_{H\subseteq G, |V(H)|\ge 2} \frac{|E(H)|}{|V(H)|-1} \Bigr\rceil.
\end{align}

This has a few consequences on the size of graphs with certain arboricity. E.g., cliques have relatively few edges for a given arboricity parameter.

\begin{lemma}
    \label{lem:arb_clique}
    The complete $n$-node graph $K_n$ with $n \geq 2$ has arboricity $\lceil \tfrac{n}{2} \rceil$.
\end{lemma}

\begin{proof}
    The densest subgraph of $K_n$ is the graph itself, i.e., it maximizes the expression $\frac{|E(H)|}{n-1}$ in Equation~\eqref{eq:arb-def2}. The number of edges of $K_n$ is $\binom{n}{2}$, therefore $a(K_n) = \lceil \binom{n}{2}/(n-1) \rceil = \lceil \tfrac{n}{2} \rceil$.
\end{proof}

This implies that $K_{2a-1}$ has arboricity $a$ and $(2a-1)(a-1)$ edges. In general, the required minimum size of graph with arboricity $a$ is close to that.

\begin{lemma}
    \label{lem:arb_min_edges}
    A graph $G$ has at least $2 \arb(G)-1$ nodes and $2(\arb(G)-1)^2 +1$ edges.
\end{lemma}

\begin{proof}
    Let $a := \arb(G)$. By Equation \eqref{eq:arb-def1} there must a subgraph $H$ of $G$ with $|E(H)| > (a-1)(|V(H)|-1) = (a-1)(\ell-1)$, where $\ell := |V(H)|$. Furthermore, $H$ can have at most $|E(H)| \leq \binom{\ell}{2} = \ell(\ell-1)/2$ edges. Thus
    \[
        \ell(\ell-1)/2 > (a-1)(\ell-1) \quad \Longleftrightarrow \quad \ell > 2(a-1) \quad \Longleftrightarrow \quad \ell \geq  2(a-1) + 1
    \]
    Therefore, $E(H) > (a-1)(\ell-1) \geq 2(a-1)^2$ edges, i.e., $H$ has at least $2(a-1)^2 +1$ edges.
\end{proof}

The lower bound in Lemma \ref{lem:arb_min_edges} is tight, i.e., the $K_{2a-1}$ clique does in fact not minimize the number of edges of an arboricity $a$ graph.

\begin{lemma}
    For any $a\geq 1$, there is a graph $G^*$ with $\arb(G^*)=a$, $2a-1$ nodes and $2(a-1)^2 +1$ edges.
\end{lemma}

\begin{proof}
    For $a=1$ $G^*$ is a single edge on two vertices. For $a \geq 2$ we let $G^{*}$ be the $K_{2a-1}$, where we remove $a-2$ edges. Then 
    \[
        |E(G^{*})| =\tbinom{2a-1}{2}-(a-2)
   =(2a-1)(a-1)-(a-2) = 2(a-1)^{2}+1.
    \]
    By Equation \eqref{eq:arb-def2}, we have $\lceil |E(G^{*})|/2(a-1)\rceil = \lceil (a-1) + \frac{1}{2(a-1)}\rceil = a$, thus $a(G^{*})= a$.
\end{proof}

A combinatorial insight by \cite{Rivin2002} shows that one can lower bound the number of distinct edges that are involved in a given set of distinct triangles.

\begin{lemma}[cf. \cite{Rivin2002} Eq.\ (8)]
    \label{lem:edges_of_triangles}
    A graph $G$ with $k$ triangles has at least $\tfrac{\sqrt{2}}{3}k^{2/3}$ edges.
\end{lemma}

\end{document}